\documentclass[11pt]{article}
\usepackage{authblk}
\usepackage{pgf}

\input{requirements}
\definecolor{royalBlue}{HTML}{057DCD}
\definecolor{darkGreen}{HTML}{2E8B57}
\definecolor{mgreen}{RGB}{160, 200, 140}
\hypersetup{
	colorlinks = true,
	linkcolor = royalBlue,
	citecolor = darkGreen,
	linktocpage = true,
	urlcolor = mgreen
}
\def\[#1\]{\begin{align*}#1\end{align*}}

\newcommand{\OFTRL}{\texttt{OFTRL}\xspace}
\newcommand{\FTRL}{\texttt{FTRL}\xspace}

\newcommand{\LRLOFTRL}{\texttt{LRL-OFTRL}\xspace}
\newcommand{\ours}{\texttt{DLRC-OMWU}\xspace}
\newcommand{\kours}{\texttt{KDLRC-OMWU}\xspace}
\newcommand{\oursgeneral}{\texttt{DLRC-OFTRL}\xspace}
\newcommand{\Opthedge}{\texttt{OMWU}\xspace}
\newcommand{\Hedge}{\texttt{MWU}\xspace}

\input{figs/tikzmacros.tex}

\newcounter{qst}
\crefname{qst}{Question}{Questions}

\numberwithin{equation}{section}

\def\ind{k}

\DeclareMathOperator{\reg}{Reg}
\DeclareMathOperator{\tildereg}{{\tilde R}eg}

\newcommand{\regdep}{\mathfrak{R}}

\DeclareMathOperator{\polylog}{polylog}

\newcommand{\defeq}{\coloneqq}
\renewcommand{\^}[1]{^{(#1)}}
\newcommand{\bbR}{\mathbb{R}}
\newcommand{\bbN}{\mathbb{N}}

\renewcommand{\vec}[1]{\bm{#1}}
\newcommand{\vstack}[2]{\begin{pmatrix} #1 \\ #2 \end{pmatrix}}

\let\hat\widehat
\let\tilde\widetilde

\DeclareMathOperator*{\dprime}{\prime \prime}

\newcommand{\argmin}{\mathop{\mathrm{arg\,min}}}
\newcommand{\argmax}{\mathop{\mathrm{arg\,max}}}

\newcommand{\kl}[2]{\textup{KL}(#1 \| #2)}
\newcommand{\ent}[1]{\textup{H}(#1)}

\newcommand{\expect}{\mathbb{E}}

\newcommand{\ut}{\vec{u}}
\newcommand{\Ut}{\vec{U}}

\newcommand{\nut}{\vec{\nu}}
\newcommand{\mut}{\vec{\mu}}
\newcommand{\at}{\vec {\mathsf{r}}}
\newcommand{\aprimet}{\vec{\mathsf{r}}^\prime}
\newcommand{\Nut}{\vec{ \mathcal{V}}}
\newcommand{\At}{\vec {\mathcal{R}}}

\newcommand{\vebar}{\overline{\vec{e}}}
\newcommand{\vb}{\vec{b}}
\newcommand{\vv}{\vec{v}}
\newcommand{\vchi}{\vec{\chi}}
\newcommand{\vx}{\vec{x}}
\newcommand{\vy}{\vec{y}}
\newcommand{\vz}{\vec{z}}
\newcommand{\vtheta}{\vec{\theta}}

\newcommand{\summ}[1]{\Lambda (#1)}

\usepackage{amsthm}

\newtheorem{theorem}{Theorem}
\newtheorem{lemma}[theorem]{Lemma}
\newtheorem{proposition}[theorem]{Proposition}
\newtheorem{corollary}[theorem]{Corollary}

\newenvironment{restate}[1]{%
  \renewcommand{\thetheorem}{\ref{#1}}%
  \addtocounter{theorem}{-1}%
  \begin{theorem}
    }{%
  \end{theorem}
  \addtocounter{theorem}{1}%
}

\newenvironment{restatelemma}[1]{%
  \renewcommand{\thelemma}{\ref{#1}}%
  \addtocounter{lemma}{-1}%
  \begin{lemma}
    }{%
  \end{lemma}
  \addtocounter{lemma}{1}%
}

\newenvironment{restateproposition}[1]{%
  \renewcommand{\theproposition}{\ref{#1}}%
  \addtocounter{proposition}{-1}%
  \begin{proposition}
    }{%
  \end{proposition}
  \addtocounter{proposition}{1}%
}

\theoremstyle{definition}
\newtheorem{definition}[theorem]{Definition}
\newtheorem{assumption}[theorem]{Assumption}
\newtheorem{observation}[theorem]{Observation}
\newtheorem{remark}[theorem]{Remark}

\numberwithin{theorem}{section}
\numberwithin{observation}{section}
\numberwithin{lemma}{section}
\numberwithin{proposition}{section}
\numberwithin{corollary}{section}
\numberwithin{remark}{section}

\usepackage{stmaryrd}

\usepackage{xcolor}

\usepackage[linesnumbered,ruled,lined,noend]{algorithm2e}

\NewDocumentCommand{\numberthis}{om}{%
  \IfNoValueTF{#1}{%
    \refstepcounter{equation}\tag{\theequation}%
  }{%
    \tag{#1}%
  }%
  \label{#2}%
}

\definecolor{darkgrey}{gray}{0.3}
\definecolor{commentcolor}{gray}{0.5}
\SetKwComment{Hline}{}{\vspace{-3mm}\textcolor{gray}{\hrule}\vspace{1mm}}
\SetKwComment{Comment}{\color{commentcolor}[$\triangleright$\ }{}
\SetCommentSty{}
\SetNlSty{}{\color{darkgrey}}{}
\SetKwProg{Fn}{function}{}{}
\SetKwProg{Subr}{subroutine}{}{}
\crefalias{AlgoLine}{line}%
\crefname{algocf}{Algorithm}{Algorithms}

\usepackage[style=alphabetic,natbib=true,maxcitenames=10]{biblatex}
\usepackage{fullpage}

\title{Faster Rates for No-Regret Learning in General Games\\ via Cautious Optimism}

\author[1]{Ashkan Soleymani}
\author[2]{Georgios Piliouras}
\author[1]{Gabriele Farina}
\affil[1]{MIT EECS, \texttt{\{ashkanso,gfarina\}@mit.edu}}
\affil[2]{Google DeepMind, \texttt{gpil@google.com}}

\date{}
\usetikzlibrary{shapes,arrows}

\begin{document}

\pagenumbering{gobble} 

\maketitle

\begin{abstract}
    We establish the first uncoupled learning algorithm that attains $O(n \log^2 d \log T)$ per-player regret in multi-player general-sum games, where $n$ is the number of players, $d$ is the number of actions available to each player, and $T$ is the number of repetitions of the game. Our results exponentially improve the dependence on $d$ compared to the $O(n\, d \log T)$ regret attainable by Log-Regularized Lifted Optimistic FTRL~\citep{farina2022near}, and also reduce the dependence on the number of iterations $T$ from $\log^4 T$ to $\log T$ compared to Optimistic Hedge, the previously well-studied algorithm with $O(n \log d \log^4 T)$ regret~\citep{daskalakis2021near}. Our algorithm is obtained by combining the classic Optimistic Multiplicative Weights Update (OMWU) with an adaptive, non-monotonic learning rate that paces the learning process of the players, making them more cautious when their regret becomes too negative.
\end{abstract}

\vspace{1cm}
\tableofcontents
\newpage
\pagenumbering{arabic} 
\section{Introduction}

The study of how multiple interacting agents learn to adapt their strategies is a well-established problem with significant foundations in game theory, online optimization, control theory, economics, and behavioral sciences~\citep{Cesa-Bianchi06:Prediction,Nisan:2007:AGT:1296179,marden2015game,gintis2000game,camerer2011behavioral}. This area has gained increased importance with the advent of machine learning, where multi-agent games are integral to many fundamental architectures and applications~\citep{goodfellow2014generative,silver2017mastering,moravvcik2017deepstack,brown2019superhuman,bighashdel2024policy}. Despite the intuitive appeal of simple uncoupled learning algorithms that globally converge to Nash equilibria in all games, strong impossibility results demonstrate that such algorithms do not exist—even in the classic and relatively constrained context of normal-form games~\citep{Hart03:Uncoupled,milionis2023impossibility}. From the perspective of centralized algorithms, \citet{daskalakis2006complexity} took a step further by linking the computational complexity of finding a Nash equilibrium in games to solving Brouwer's fixed-point problem, showing that computing a Nash equilibrium in polynomial time is impossible unless PPAD = P, suggesting the problem is inherently computationally difficult.

These roadblocks have inspired the pursuit of alternative solution concepts, with the notion of \textit{regret minimization} being arguably the most influential and well-studied among them~\citep{Shalev-Shwartz12:Online,hazan2016introduction}. Originally defined within the context of single-agent optimization, regret measures the difference between the accumulated performance of an algorithm and that of the best fixed action in hindsight, under uncertain and possibly adversarial realizations of payoffs. Regret minimization in general games is sufficient to establish the time-average convergence of the empirical distribution of play to the set of Coarse Correlated Equilibria (CCE). This set is a relaxation of the Nash equilibrium concept—motivated by the intractability of Nash equilibrium—and possesses several desirable properties, including approximate optimality of the resulting social welfare, connections to Nash and correlated equilibria, broad applicability across a diverse range of games, and flexibility in strategy coordination~\citep{Blum06:Routing,nadav2010no,Roughgarden15:Intrinsic,cai2016zero,monnot2017limits,roughgarden2017price}. Beyond convergence to CCE in game settings, regret minimization is recognized as a fundamental concept with wide-ranging applications, including learning and generalization~\citep{littlestone1988learning,blum1990learning,lugosi2023online}, von Neumann’s minimax theorem and Blackwell approachability~\citep{freund1996game,abernethy2011blackwell}, boosting~\citep{freund1995desicion,freund1996game}, combinatorial optimization~\citep{plotkin1995fast,arora2007combinatorial}, complexity theory~\citep{klivans1999boosting,barak2009uniform}, differential privacy~\citep{hardt2010multiplicative,hsu2013differential}, prediction markets~\citep{chen2010new,abernethy2013efficient}, evolutionary dynamics~\citep{chastain2013multiplicative}, and more.

In the adversarial case, it is classically known that simple algorithms such as Hedge/Multiplicative Weights Update (MWU) suffice to establish the optimal regret rate of $O(T^{1/2})$~\citep{Freund97:decision}. Nevertheless, an adversarial framework is often overly pessimistic for applications in game theory, as it may yield suboptimal results in more favorable and predictable environments. This is particularly evident in the context of learning in games, where player interactions are nonstationary but typically evolve slowly. Consequently, determining the tightest possible bounds for no-regret learning in general games remains a fundamental and unresolved problem. We enumerate major previous attempts in this line of work in \Cref{sec:prev_me}.

By the \emph{Optimism} framework of \citet{Syrgkanis15:Fast}, there exist algorithms for which the sum of players' regrets in self-play remains constant over time. Unfortunately, controlling individual regrets—necessary for convergence to Coarse Correlated Equilibria (CCE)—is \emph{much harder}, and despite nearly a decade of noteworthy progress, the optimal regret rate remains elusive. Intuitively, we seek to avoid agents with runaway negative regret. Our main high-level idea is to adapt optimism into a form of \emph{Cautious Optimism}, where agents decrease their learning rates when their regret becomes too negative—that is, when they significantly outperform all fixed actions. Surprisingly, we show that it is possible to carefully apply this germ of an idea to achieve new state-of-the-art regret bounds.

\subsection{Overview of the Main Result and Techniques} \label{sec:cont}

We establish the first \emph{uncoupled} no-regret learning algorithm that attains $O(n \log^2 d \log T)$ regret in multi-player general-sum games. Our algorithm is best understood as an Optimistic Multiplicative Weights Update (\Opthedge) paired with a \textit{dynamic} learning rate that is adjusted based on the regret accumulated by the learner. For this reason, we coin our method \emph{Dynamic Learning Rate Control OMWU (\ours)}. A key characteristic that sets our approach apart from standard adaptive learning-rate techniques is that our dynamic control does not produce monotonically decreasing learning rates. Instead, the goal of our dynamic learning rate is to properly pace the learner—\emph{slowing it down when it is performing too well}—that is, when its maximum regret becomes too negative. We achieve this by solving an optimization problem at each iteration to determine the learning rate, based on the player's current regret vector.

The idea of agents differentiating their behavior depending on whether they feel content or discontent is both simple and intuitive, and has inspired other game dynamics that provably concentrate around pure Nash equilibria~\citep{young2009learning}, as well as adjusted replicator dynamics in evolutionary game theory~\citep{weibull1997evolutionary}. In the context of regret minimization, this idea can be traced back to the work of \citet{bowling2002multiagent}, who introduced the Win or Learn Fast (WoLF) principle. This principle increases the learning rate of agents when they are losing, thereby resulting in quicker adaptation to the environment and to the strategies of other agents. \citet{bowling2002multiagent} demonstrated the convergence of gradient ascent-descent with WoLF to a Nash equilibrium in the restricted case of two players with two actions. \citet{bowling2004convergence} extended WoLF to multiplayer settings, showing it achieves no-regret dynamics with $O(d \sqrt{T})$ regret. Despite the success of WoLF and related heuristics in small-scale games~\citep{abdallah2008multiagent,banerjee2003adaptive,bloembergen2015evolutionary,kaisers2009evolutionary,leonardos2022exploration}, their theoretical guarantees in general games remain unclear.

On the negative side, it has recently been shown that the multiplicative weights update algorithm, even when equipped with a continuous learning rate (rather than the original fast and slow rates of WoLF~\citep{bowling2002multiagent}), exhibits chaotic behavior in nonatomic congestion games~\citep{vlatakis2023chaos}. While our motivation for adaptive learning rates and our methodology differ from the WoLF principle, to our knowledge, \ours is the \emph{first algorithm to demonstrate theoretical benefits of such ideas in regret minimization for general games}.

The regret guarantees of \ours exponentially improve the dependence on $d$ in the $O(n d \log T)$ regret attained by Log-Regularized Lifted Optimistic FTRL~\citep{farina2022near}. Compared to the regret analysis of \Opthedge in \citet{daskalakis2021near}, we improve the regret bound in several aspects. Not only do we reduce the dependence on the time horizon $T$ from $\log^4 T$ to $\log T$, but our guarantees also hold across all regimes of $n$, $d$, and $T$, whereas theirs require $T > C n \log^4 T$ for some constant $C$ (see \Cref{sec:regret_analysis}, and Lemmas 4.2 and C.4 in~\citep{daskalakis2021near}). Additionally, as a minor note, the regret analysis of \Opthedge~\citep{daskalakis2021near} hides extremely large constants in the asymptotic notation, especially in comparison to our analysis.

\paragraph{Technical Contributions.} 
The technical novelties of this work, which lead to fast no-regret learning rates, are multifaceted and can be briefly summarized as follows.

First, as discussed previously, we conceptualize the idea of learning rate control for no-regret learning (see \Cref{sec:loga,sec:design_dynamic_learning}) and formalize it to design optimization algorithms for learning rate control, resulting in \ours, a computationally efficient algorithm (see \Cref{sec:loga}). The concept of dynamic learning rates has the potential to be beneficial in other areas involving regret minimization, particularly in multi-agent settings.

\begin{table}[t]%
     \newcommand{\ldarrow}{\raisebox{-.7mm}{\tikz \draw[->] (0,0) -- (.25,0) -- +(0, -.2);}}%
    \scalebox{0.9}{%
    \begin{tabular}{@{}>{\arraybackslash}m{4.0cm} >{\arraybackslash}m{4.0cm} >{\arraybackslash}m{5.3cm} c@{}} %
        \bf Method                                            & \bf Games' Regret     & \bf Cost per Iteration  & \bf Adversarial Regret  \\
        \toprule
        OFTRL / OMD\newline\citep{Syrgkanis15:Fast}            & $O(\sqrt{n}\, \regdep T^{1/4})$   & Regularizer- \& oracle-dependent &  $\Tilde{O}(\sqrt{T \log d})$ \\
        \midrule
        OMWU\newline\citep{Chen20:Hedging}\!                             & $O(n \log^{5/6} d \: T^{1/6})\: \dagger $                     & $O(d) $  &    $\Tilde{O}(\sqrt{T \log d})$ \\ 
        \midrule
        OMWU\newline\citep{daskalakis2021near}\!                             & $O(n \log d \log^4 T)$                     & $O(d)$   &   $\Tilde{O}(\sqrt{T \log d})$      \\
        \midrule
        Clairvoyant MWU\newline\citep{Piliouras22:Optimal}              & $O(n \log d)$\newline for a subsequence only~$\ddagger$      & $O(d)$   & No guarantees                                                               \\
        \midrule
          LRL-OFTRL\newline\citep{farina2022near}              &  $O(n \hspace{0.5 mm} d  \log T)$     & $O(d \log \log T)$  &   $\Tilde{O}(\sqrt{T \log d})$                                                                            \\
        \midrule
          \ours\newline\textbf{[This paper]}             &  $O(n \log^2 d \log T)$                  &       $O(d \log \log T)$   &   $\Tilde{O}(\sqrt{T \log d})$                                                             \\
       \bottomrule
    \end{tabular}
    }
    \caption{Comparison of prior results on minimizing external regret in general games. For simplicity, we omit dependencies on the smoothness and range of the utility functions. We use $n$ to denote the number of players, $T$ the number of repetitions of the game, and $d$ the number of actions. $\regdep$ denotes the maximum value attained by the regularizer. $\dagger$ Limited to two-player games only ($n = 2$). $\ddagger$ Unlike all other algorithms, the full sequence of iterates produced by Clairvoyant MWU (CMWU) is not known to achieve sublinear regret. Instead, after running CMWU for $T$ iterations, only a smaller subsequence of length $\Theta(T / \log T)$ is known to attain the regret stated in the table.}
    \label{table:results}
\end{table}

Secondly, using analysis techniques based on the properties of self-concordant functions, we establish strong \emph{multiplicative stability} results for the resulting learning rate, even when the actions themselves are not known to be multiplicatively stable (see \Cref{sec:learning_rate_control_problem,sec:proof_temp_adjust}). This approach enables predictability of the dynamics beyond the constant learning rate setting. Previously, it was unclear how to design learning dynamics that evolve smoothly while simultaneously adapting to changes induced by the learning processes of other agents.

Thirdly, we demonstrate that combining the Optimistic Multiplicative Weights Update algorithm with our learning rate control can be viewed as an instantiation of optimistic Follow-the-Regularized-Leader (FTRL) with a novel regularizer $\psi$. We further explore equivalent viewpoints of \ours, revealing connections to the Lifted Optimistic FTRL algorithm proposed by \citet{farina2022near}, albeit with a different regularizer (see \Cref{sec:design_dynamic_learning,sec:equivalent_viewpoints}), where the FTRL dynamics are executed over the lifted space.

Fourthly, we introduce a novel regularizer, $\psi$, whose spectral properties lie between those of the logarithmic and entropic regularizers. We establish \emph{strong spectral properties} for this regularizer, including strong convexity and high curvature (see \Cref{sec:proof_sketch,sec:spectral_psi}). This approach allows us to depart from the traditional methodology of using intrinsic norms induced by the Hessian matrix, which is prevalent in the online learning literature. We note that the regularizer $\psi$ may be of independent interest in future studies.

Lastly, it is important to mention that, in \Cref{sec:kernel}, inspired by Kernelized \Opthedge~\citep{farina2022kernelized}, we introduce a kernelized version of \ours, (denoted as \kours), extended to convex $0$/$1$-polyhedral games such as extensive-form games or flows on directed graphs. This way, we demonstrate that \ours inherits the fundamental and intriguing properties of \Opthedge.

\smallskip

While we adopt the intriguing nonnegative RVU property idea from Log-Regularized Lifted \OFTRL (\LRLOFTRL)~\citep{farina2022near} (see \Cref{sec:regret_analysis,sec:pos_regret}), our construction and proof technique differ significantly. The analysis in \LRLOFTRL relies heavily on the intrinsic norm of the logarithmic regularizer to ensure multiplicative stability of the iterates. In contrast, our approach does not depend on multiplicative stability or intrinsic norms. Instead, it involves a more nuanced analysis, focusing on the dynamic learning rate control optimization problem and the strong spectral properties of our specifically chosen regularizer.

\subsection{Prior Work on Regularized Learning in General Games} \label{sec:prev_me}

Starting with the seminal paper of \citet{Syrgkanis15:Fast}, which established the first $o(T^{1/2})$ regret bound for self-play in general games, a race was initiated to provide the tightest possible bounds in this setting. In particular, \citet{Syrgkanis15:Fast} identified the \emph{RVU property}, an adversarial regret bound applicable to a broad class of so-called \emph{Optimistic} variants of standard no-regret learning algorithms, such as Follow-the-Regularized-Leader and Mirror Descent. Using this property, they demonstrated that the individual regret of each player grows as $O(T^{1/4})$.

In a more recent breakthrough, \citet{daskalakis2021near} exponentially improved the regret guarantees for Optimistic MWU (OMWU), achieving the first polylogarithmic in $T$ regret bound of $O(n \log d \log^4 T)$, where $n$ is the number of players and $d$ is the number of actions per player. Their analysis relies heavily on discrete-time Fourier transforms and the smoothness of higher-order discrete differentials of the learning dynamics in the frequency domain.

Clairvoyant MWU (CMWU) advanced the concept of prediction even further~\citep{Piliouras22:Optimal}, albeit at the cost of creating an algorithm that does not minimize regret in adversarial settings. The key intuition is that, given a perfect prediction of future payoff sequences, a bounded total payoff can be easily achieved by applying a Be-the-Leader type of algorithm. Although such approaches are not feasible in standard online learning, agents can implement such a sequence of play in game settings via uncoupled learning algorithms. The full sequence of iterates produced by CMWU does not achieve sublinear regret. Instead, after running CMWU for $T$ iterations, only a smaller subsequence of length $\Theta(T / \log T)$ maintains a bounded total regret of $O(n \log d)$. Combining these results, the effective convergence rate toward CCE is $O\left(\frac{n \log d \log T}{T}\right)$. For a more detailed discussion of how CMWU differs from other algorithms in this line of work, see \Cref{sec:related_work}.

In the last major step prior to our work, \citet{farina2022near} achieved logarithmic dependence on $T$ over the entire history of play via the Log-Regularized Lifted Optimistic FTRL algorithm. Their learning dynamics are based on an instantiation of Optimistic Follow-the-Regularized-Leader over an appropriately \emph{lifted} space using a logarithmic regularizer. Interestingly, their approach generalizes beyond normal-form games to encompass general convex games. However, this analysis comes at the cost of exponentially worse dependence on the number of actions $d$, with an overall regret guarantee of $O(n d \log T)$.

\Cref{table:results} provides a summary of prior work aimed at establishing optimal regret bounds for no-regret learners in games.

\subsection{Additional Related Work} \label{sec:related_work}

The evolution of regret minimization in games has seen significant breakthroughs, beginning with the pioneering work of \citet{Daskalakis11:Near} on zero-sum games. They developed \emph{strongly uncoupled} learning dynamics that achieve a regret growth rate of $O(\log T)$. This milestone was further refined by \citet{Rakhlin13:Optimization}, who introduced \emph{Optimistic Mirror Descent} (OMD), simplifying the implementation while maintaining robust performance. The practical benefits of recency bias in OMD have also been substantiated in behavioral economics~\citep{Fudenberg14:Recency}.

Beyond their near-optimal performance, an additional advantage of optimistic learning algorithms is their incorporation of a straightforward and intuitive recency bias, which aligns well with real-world behavioral biases~\citep{nevo2012surprise,erev2013learning}. Each agent optimistically assumes that all other agents will play tomorrow exactly as they did today, and then applies their online learning algorithm, incorporating this extrapolation step into the decision-making process. This alignment with behavioral data has two key benefits. First, it suggests that the resulting theoretical analysis may offer insights into human behavior in everyday interactions. Second, it highlights the potential for behavioral science to inspire algorithmic modifications with even stronger performance guarantees.

\citet{farina2022kernelized} extended these theoretical advancements beyond normal-form games, generalizing the $O(\polylog(T))$ regret bounds to polyhedral games, which encompass extensive-form games. The introduction of \emph{no-swap-regret} dynamics has also led to fast convergence to correlated equilibria in normal-form games~\citep{Chen20:Hedging,Anagnostides21:Near,Anagnostides22:Uncoupled}. \citet{wei_more_2018} further advanced the field by leveraging optimism to achieve adaptive regret bounds in bandit settings.

The focus on \emph{last-iterate} convergence has brought renewed interest, with significant contributions highlighting the performance of Optimistic Mirror Descent (OMD) algorithms~\citep{daskalakis_last-iterate_2019,mertikopoulos2019optimistic,golowich_tight_2020,lei_last_2021,hsieh_adaptive_2021,wei_linear_2021,azizian_last-iterate_2021}. These works examine the conditions under which OMD algorithms achieve favorable iterate convergence properties.

Parallel to these developments, extensive research has examined the dynamics of learning algorithms in games. This includes studies that do not incorporate optimism and often reveal complex behaviors such as divergence, recurrence, or chaos~\citep{hart_uncoupled_2003,daskalakis2010learning,kleinberg_beyond_2011,balcan_weighted_2012,Entropy18,bailey_multiplicative_2018,mertikopoulos_cycles_2018,bailey_fast_2019,cheung_vortices_2019,vlatakis_no-regret_2020,bailey2020finite,cheung2020chaos}. These findings underscore the intricate and often unstable nature of learning dynamics in strategic environments, highlighting both the challenges and opportunities that remain in this vibrant area of research.

\paragraph{On Clairvoyant MWU.} As a specific uncoupled learning algorithm, Clairvoyant MWU (CMWU) was introduced by \citet{Piliouras22:Optimal} to compute the Coarse Correlated Equilibrium (CCE) of a normal-form game. It is known that, given a perfect prediction of future payoff sequences, a bounded total regret of $O(n \log d)$ can be easily achieved by applying a Be-the-Leader type algorithm. Although such approaches are not feasible in online learning, in self-play settings—under certain predetermined agreements among the players—it is possible to implement a sequence of play that is informed by hindsight of the game’s next action. 

\citet{farina2022clairvoyant} showed that CMWU is equivalent to Nemirovski’s Conceptual Proximal Method~\citep{nemirovski2004prox}. Given the proximal operator of the game (in this setting, the online mirror descent optimization step of all players combined), if the players play the fixed point of this operator, they are guaranteed to achieve constant regret. In the setting of normal-form games, this operator is proven to be monotone for a sufficiently small choice of learning rate~\citep{Piliouras22:Optimal,farina2022clairvoyant}. Thus, its fixed point can be computed—either in a centralized or decentralized manner.

The key idea in \citet{Piliouras22:Optimal} is to use self-play communication as a broadcast channel, allowing players to collaboratively find this fixed point by broadcasting their updates through in-game actions. In this setup, for $T - O(\log T)$ iterations, players are merely communicating indirectly. Only once every predetermined $O(\log T)$ rounds do they play the approximate fixed point they have collectively computed. Consequently, only $O(\log T)$ actual gameplay iterations occur, and the empirical distribution over these predetermined rounds converges to the CCE of the game, resulting in an effective convergence rate of $O\left(\frac{n \log d \log T}{T}\right)$.

While this approach yields faster rates, it comes with two caveats:  
(i) it does not guarantee sublinear regret in adversarial settings or even if a single player deviates slightly;  
(ii) players must have prior agreements on how to self-play and must strictly adhere to the protocol. For instance, if one player begins a round early or late, the communication process fails, as the true gameplay during the $O(\log T)$ periods becomes mixed with the fixed-point computation steps of other players. Due to these issues, this algorithm is not applicable to large-scale games in real-world applications.

\section{Preliminaries}
\label{section:background}

For any $d \in \bbN$, we denote the set $\{1, 2, \dots, d\}$ by $[d]$. We use bold letters to denote vectors, for example, $\vx \in \bbR^d$. We let $\vx[r]$ denote the $r$-th coordinate of the vector $\vx$, for any $r \in [d]$. In our notation, players are generally indicated by subscripts, which we omit whenever the results involve a generic player and are clear from context, to avoid overloading the notation. The time index, represented by the variable $t$, is indicated using superscripts. For example, $\vx_i\^t$ denotes the action played by player $i$ at time $t$. We denote the set of distributions supported on a finite set $\mathcal{A}$ of size $d = |\mathcal{A}|$ by $\Delta^d \defeq \Delta(\mathcal{A})$. For given vectors $\vx$ and $\ut$, we denote their inner product by $\langle \vx, \ut \rangle$. For any vector $\vx \in \bbR^d$, we write $\summ{\vx} \defeq \sum_{\ind = 1}^{d} \vx[\ind]$ for the sum of its elements. We also define the negative entropy by $\ent{\vx} = \sum_{\ind = 1}^{d} \vx[\ind] \log \vx[\ind]$ and the Kullback–Leibler (KL) divergence by $\kl{\vx}{\vx'} = \sum_{\ind = 1}^{d} \vx[\ind] \log \frac{\vx[\ind]}{\vx'[\ind]}$. Lastly, we use the notation $\vec{1}_d$ to denote the vector $[1, 1, \dots, 1] \in \bbR^d$.

\paragraph{General-sum Games.} We consider $n$-player general-sum games and denote the set of players by $[n]$. In a normal-form game, each player $i \in [n]$ has a finite and nonempty set of deterministic strategies $\mathcal{A}_i$. The set of mixed strategies for player~$i$ is given by the probability simplex over $\mathcal{A}_i$, denoted $\mathcal{X}_i = \Delta(\mathcal{A}_i)$. The joint action space of the players is then $\bigtimes_{j=1}^n \mathcal{A}_j$. Each player~$i$ is associated with a utility function $\mathcal{U}_i : \bigtimes_{j=1}^n \mathcal{A}_j \to \bbR$ defined over joint deterministic strategies. Let the expected utility of player~$i$ under a joint mixed strategy profile $\vx = (\vx_1, \vx_2, \dots, \vx_n) \in \bigtimes_{j=1}^n \Delta(\mathcal{A}_j)$ be denoted by the multilinear extension:
\[
\nut_i(\vx) \defeq \mathbb{E}_{\vec{s} \sim \vx}[\mathcal{U}_i(\vec{s})] = \langle \vx_i, \nabla_{\vx_i} \nut_i(\vx) \rangle.
\]

Let $d \defeq \max_{i \in [n]} |\mathcal{A}_i|$ be the maximum number of actions available to any player. Without loss of generality, we assume that $|\mathcal{A}_i| \geq 2$ for each player~$i$, since players with only one action can be removed from the game without loss of generality. For simplicity of notation, we further assume that all players have exactly $d$ actions. Finally, we analyze the game under the following standard assumption regarding the boundedness and smoothness of the utility functions.

\begin{assumption} \label{assumption:bound}
    The utility function $\nut_j$ of the game (for each player $j \in [n]$) satisfies, for any two strategy profiles $\vx, \vec x^\prime \in \bigtimes_{j = 1}^{n}  \Delta (\mathcal{A}_j)$,
    \begin{itemize}
        \item (Bounded Utilities)
            $ \max_{s \in {\bigtimes_{i=1}^n \mathcal{A}_i}} |\mathcal{U}_j(s)| \leq 1 $, for each player $j \in [n]$.   
        \item ($L$-smoothness) there exists a positive number $L > 0$ such that for each player $j \in [n]$, 
            $$ \| \nabla_{\vx_j} \nut_j (\vx) - \nabla_{\vx_j} \nut_j (\vec x^\prime) \|_{\infty} \leq L \sum\limits_{i \in [n]} \|\vx_i - {\vec x^\prime_i} \|_1. $$
    \end{itemize}
\end{assumption}

This assumption is general and not restrictive, as any bounded utility function can be rescaled—without loss of generality—to satisfy the boundedness condition. Moreover, it is straightforward to show that under this assumption, the game is $L$-smooth with $L = 1$. However, depending on the structure of the game, the smoothness parameter $L$ may be substantially smaller than 1. For example, in games with vanishing sensitivity $\epsilon_n$, it can be observed that $L = \epsilon_n \ll 1$~\citep{anagnostides2024interplay}. Here, we distinguish $L$ from the boundedness assumption in order to account for additional structural properties that may lead to faster convergence guarantees.

\paragraph{No-regret Learning.}
In the \emph{online learning framework}, a learning agent is required to select a strategy $\vx\^t \in \mathcal{X} = \Delta^d$ at each time $t \in \bbN$. We consider the \emph{full-information} model, in which the environment provides feedback in the form of a \emph{linear} utility function $\vx \mapsto \langle \vx, \nut\^t \rangle$, where $\nut\^t \in \bbR^d$. The principal measure of the agent's performance is \emph{external regret} (also referred to as \emph{regret}), defined over a time horizon $T \in \bbN$ as:
\begin{equation}
\label{eq:linreg}
\reg\^T \defeq \max_{\vx^* \in \Delta^d} \left\{ \sum_{t=1}^T \langle \nut\^t, \vx^* \rangle \right\} - \sum_{t=1}^T \langle \nut\^t, \vx\^t \rangle.
\end{equation}
This metric evaluates the agent's performance against the best possible \emph{fixed} strategy chosen in hindsight. The goal is to ensure that regret grows as slowly as possible with the time horizon $T$. It is well known that standard algorithms such as Online Mirror Descent (OMD) and Follow-the-Regularized-Leader (FTRL) achieve the optimal regret rate of $\reg\^T = \Theta(\sqrt{T})$ in adversarial settings~\citep{hazan2016introduction,orabona2019modern}. However, in \emph{game self-play} settings, faster rates are achievable, since the dynamics of each player's online learning process follow the same learning algorithm, in contrast to the adversarial case. Specifically, each player $i \in [n]$ at time step $t$ receives a utility vector $\nut\^t = \nabla_{\vx_i} \nut_i(\vx\^t)$, where $\vx\^t = (\vx_1\^t, \vx_2\^t, \dots, \vx_n\^t)$ denotes the joint play at time $t$. It is important to note that a player's regret may be negative.

An important class of no-regret learning algorithms that achieve faster convergence guarantees in games are those that satisfy the \emph{Regret bounded by Variation in Utilities (RVU)} property—a concept introduced by \citet{Syrgkanis15:Fast}, which we formalize below.

\begin{definition}[RVU property~\citep{Syrgkanis15:Fast}]
    A no-regret learning algorithm satisfies Regret bounded by Variation in Utilities (RVU) property with parameters $a, b, c$ and a pair of dual norms $(\| .\|, \| . \|_{*})$, if the regret on any sequence of utilites $\{\nut\^t\}_{t = 1}^T$ is upper bounded by
     \[
        \reg\^T \leq a + b \sum_{t=1}^{T} \| \nut\^{t} - \nut\^{t-1} \|_{*}^2 - c \sum_{t=1}^{T} \| \vx\^{t+1} - \vx\^{t}  \|^2,
    \]
    where $a \geq 0$ and $0 < b \leq c$.
\end{definition}

It has been shown that optimistic variants of OMD and FTRL algorithms—originally introduced by \citet{rakhlin2013online}—satisfy the RVU property~\citep{Syrgkanis15:Fast}, and achieve faster rates in the self-play setting compared to the adversarial one. These algorithms also attain the optimal convergence rate for the \emph{sum of regrets}; however, it remains an open question whether the RVU property alone suffices to guarantee optimal individual no-regret learning rates in games.

\paragraph{Coarse Correlated Equilibrium.} A probability distribution $\sigma \in \Delta\left(\bigtimes_{j=1}^n \mathcal{A}_j\right)$ over the joint action space $\bigtimes_{j=1}^n \mathcal{A}_j$ is called an $\epsilon$-\emph{approximate coarse correlated equilibrium (CCE)} if, for every player $i \in [n]$ and every deterministic strategy $s_i' \in \mathcal{A}_i$, unilateral deviations do not increase the expected utility by more than $\epsilon$, that is:
\[
\mathbb{E}_{\mathbf{s} \sim \sigma}\left[\mathcal{U}_i(\mathbf{s})\right] \geq \mathbb{E}_{\mathbf{s} \sim \sigma}\left[\mathcal{U}_i(s_i', \mathbf{s}_{-i})\right] - \epsilon,
\]
where $\mathbf{s} = (s_1, s_2, \ldots, s_n)$ denotes the joint action drawn from $\sigma$, and $\mathbf{s}_{-i}$ represents the actions of all players other than $i$.

A folklore result in game theory establishes a connection between no-regret learning algorithms and coarse correlated equilibria (CCE) of a game~\citep{Cesa-Bianchi06:Prediction}. When the players follow no-regret learning algorithms with regrets $\reg_i\^T$ for all $i \in [n]$, the empirical play of the game, defined as $\sigma = \frac{1}{T} \sum_{t=1}^{T} \vx\^t$, constitutes an $\left(\frac{1}{T} \max_{i \in [n]} \reg_i\^T\right)$-approximate CCE of the underlying game. Hence, faster regret rates in self-play settings lead directly to faster convergence to the game's coarse correlated equilibrium.

\section{Dynamic Learning Rate Control} \label{sec:loga}

We restate that our algorithm is a variant of the Optimistic Multiplicative Weights Update algorithm (\Opthedge), in which the learning rate is selected adaptively based on the regret accumulated up to any point in time. For this reason, we refer to our method as \emph{Dynamic Learning Rate Control \Opthedge (\ours)}.

In its standard version, \Opthedge picks the next distribution of play proportionally to the exponential of the optimistic regret accumulated on each action, that is, according to the formula
\[
    \vx\^{t}[k] \defeq\frac{\exp\{\lambda\^t \at\^t[k]\}}{\sum_{k' \in \mathcal{A}} \exp\{\lambda\^t \at\^t[k']\}} \qquad \forall k \in \mathcal{A},
    \numberthis{eq:softmax}
\]
where $\lambda\^t > 0$ is a learning rate and $\at\^t$ is the vector of optimistically-corrected regrets accumulated by each action up to time $t$,
\[
    \at\^t[k] \defeq (\nut\^{t-1}[k] - \langle \nut\^{t-1}, \vx\^{t-1}\rangle) + \sum_{\tau=1}^{t-1} \left[\nut\^\tau[k] - \langle \nut\^\tau, \vx\^\tau\rangle\right] \qquad \forall k \in \mathcal{A}.
\]
For simplicity, define the corrected reward signal $\ut\^t \defeq \nut\^t - \langle \nut\^t, \vx\^t\rangle \vec{1}_d$ and the accumulated signal $\Ut\^t \defeq \sum_{\tau=1}^{t-1} \left[\nut\^\tau - \langle \nut\^\tau, \vx\^\tau\rangle \vec{1}_d\right]$.

It was recently discovered that when all players in an $n$-player game learn using \Opthedge with a suitable \emph{constant} learning rate $\lambda\^t = \eta$, the maximum regret accumulated by the players grows at most as $O(n \log d \log^4 T)$ as a function of the time horizon, for sufficiently large $T$~\citep{daskalakis2021near}.

In this paper, we show that the previous result can be improved to a $\log T$ dependence by using a different approach based on a novel technique that we term \emph{dynamic learning rate control}. Unlike most of the literature on adaptive learning rate schedules in online learning and optimization, our dynamic control does not produce monotonically decreasing learning rates and conceptually, it is \emph{not} designed as a means of circumventing uncertainty about the problem’s conditioning (e.g., the Lipschitz constants). Rather, our dynamic learning rate control aims to pace the learner—\emph{slowing it down when it is performing too well}—that is, \emph{when its maximum regret is too negative}.

Such a goal might appear backwards---after all, if a learner is doing so well, why pace them down? The contradiction is resolved when considering the learning system as a whole. Intuitively, if one player is doing too well, the other players might be unable to catch up. Instead, when all players' regrets remain nonnegative, one is able to show desirable overall properties of the learning process. These include not only small swap regret~\citep{Anagnostides22:Uncoupled}, but also iterate convergence to equilibrium~\citep{Anagnostides22:Last-Iterate}, and discovery of strongly incentive-compatible equilibria~\citep{Anagnostides22:Optimistic}. Alternatively, there is another way to look at this idea. As discussed in \Cref{section:background}, no-regret dynamics lead to convergence to coarse correlated equilibria at a rate determined by the \emph{worst-performing} player---the player with the maximum regret. Therefore, it seems that for faster convergence, enforcing some form of performance balance among the players of the game is natural.

We present our algorithm, \emph{Dynamic Learning Rate Control Optimistic MWU (\ours)}, in \Cref{algo:ours_lambdaview}. We prove that when the players of the game follow \ours (i.e., in the self-play setting), each player experiences $O(n \log^2 d \log T)$ regret, while in the adversarial setting, each player suffers at most $\tilde{O}(\sqrt{T \log d})$ regret. Therefore, \ours is robust to adversarial behavior. The dynamics of \ours are simple and mirror those of \Opthedge, except that when the maximum regret becomes too negative (Line 6 in \Cref{algo:ours_lambdaview}), the learning rate in the Multiplicative Weights Update step is adjusted dynamically (Line 7 in \Cref{algo:ours_lambdaview}). Our theoretical guarantees in \Cref{theorem:informal_regret_bound} hold across all regimes of $n$, $d$, and $T$, as a constant choice of $\eta$ is sufficient for the result.

\begin{theorem}[Informal; see \Cref{theorem:regret_bound} for the detailed version] \label{theorem:informal_regret_bound}
Suppose that $n$ players self-play a general-sum multiplayer game with a finite set of $d$ deterministic strategies per player over $T$ rounds. Further, suppose that each player follows \ours to choose their action based on the history so far. Then, each player incurs $O(n \log^2 d \log T)$ regret. Moreover, when faced with adversarial utilities, each player playing \ours is guaranteed to experience $\tilde{O}(\sqrt{T \log d})$ regret.

\end{theorem}

An immediate consequence of \Cref{theorem:informal_regret_bound} is that the empirical distribution of play constitutes an $O\left(\frac{n \log^2 d \log T}{T}\right)$-approximate coarse correlated equilibrium (CCE) of the game. This corollary follows from the fact that, in a general-sum multiplayer game, if each player's regret is at most $\epsilon(T)$, then the empirical distribution of their joint strategies converges to a CCE at a rate of $O(\epsilon(T)/T)$.

\begin{corollary}
    If $n$ players employ the uncoupled learning dynamics of \ours for $T$ rounds in a general-sum multiplayer game with a finite set of $d$ deterministic strategies per player, then the empirical distribution of play forms an $O\left(\frac{n \log^2 d \log T}{T}\right)$-approximate coarse correlated equilibrium (CCE) of the game.
\end{corollary}

\begin{algorithm}[htp]
    \SetNoFillComment
    \caption{Dynamic Learning Rate Control - Optimistic MWU (\ours)
    }\label{algo:ours_lambdaview}
    \DontPrintSemicolon
    \KwData{Learning rate $\eta$, parameters $\alpha$ and $\beta$}\vspace{2mm}
    Set $ {\Ut}\^1, \ut^{(0)} \gets \vec{0} \in \bbR^{d}$\;
    \For{$t=1,2,\dots, T$}{
    Set $\at\^t \gets {\Ut}\^t + {\vec\ut}\^{t-1}$ \Comment*{\color{commentcolor}Optimism]\!\!\!\!\!}
    \uIf{$\max_{\ind \in [d]} \{ \at[\ind]\} \geq - \beta \log^2 d$}{
    Set  $\displaystyle \lambda\^t \gets \eta$ \;
    }
    \Else{
    \tcc{\color{commentcolor}\texttt{Dynamic Learning Rate Control}}
     Set  $\displaystyle \lambda\^t \gets \argmax_{\lambda \in (0, \eta]} \left\{ \log \Big( {\sum_{\ind=1}^d e^{\lambda \at\^t[\ind]}} \Big) + (\alpha - 1) \log \lambda  \right\}$ \medskip \label{line:oftrl0} 
    }
    Play strategy $\displaystyle\vx\^t[\ind] \defeq \frac{\exp\{\lambda\^t \at\^t[\ind]\}}{\sum_{\ind' = 1}^{d} \exp\{\lambda\^t \at\^t[\ind']\}} $\Comment*{\color{commentcolor} \Opthedge{}]\!\!\!\!\!} \label{line:norm0}
    Observe $\vec \nu \^t \in \bbR^d$\;
    Set $\displaystyle {\vec\ut}\^t \gets \vec \nu \^t -\langle \nut\^t, \vx\^t\rangle \vec1_d $ \;
    \label{line:lift0}
    Set $ {\Ut}\^{t+1} \gets  {\Ut}\^t +  {\vec\ut}\^t$
    \Comment*{\color{commentcolor}Empirical cumulated regrets]\!\!\!\!\!}
    }
\end{algorithm}

\subsection{The Learning Rate Control Problem} \label{sec:learning_rate_control_problem}

At every time $t$, our learning algorithm outputs strategies according to \eqref{eq:softmax}, where the learning rate is carefully chosen as the solution to the following optimization problem:
\[
    \lambda\^t \defeq \argmax_{\lambda \in (0, \eta]} \left\{f(\lambda; \at\^t) \defeq (\alpha - 1) \log \lambda + \log \left( \sum_{\ind=1}^d e^{\lambda \at\^t[\ind]} \right) \right\}, \numberthis{eq:opt_problem_lambda}
\]
where $\eta > 0$ is a constant that caps the maximum learning rate, and $\alpha$ is a key parameter of the algorithm. We set $\alpha$ to be on the order of $\Theta(\log^2 d)$ to achieve the guarantees mentioned in the introduction.

Under normal operating conditions, where the maximum regret $\max_k \{\at\^t[k]\}$ accumulated on the actions is not too negative, it is immediate to observe that the optimal solution is $\lambda\^t = \eta$,\footnote{Please refer to \Cref{lemma:lambda_is_one} for a concrete proof.} thus recovering the usual operating regime of a constant learning rate. However, when the maximum regret becomes sufficiently negative, the optimal value of $\lambda\^t$ starts decreasing towards $0$, causing the corresponding player to degrade performance by disregarding the history of the game. In the extreme case where $\lambda\^t \rightarrow 0$, the player acts uniformly among its actions. \Cref{fig:lambda} illustrates the value of the learning rate $\lambda\^t$ as a function of the optimistic cumulative regrets in a simple two-action case. It can be observed that $\lambda\^t$ is a monotonically non-increasing function of $\{\at\^t[1]\}$ and $\{\at\^t[2]\}$.

\begin{figure}[htp]
\scalebox{.9}{\input{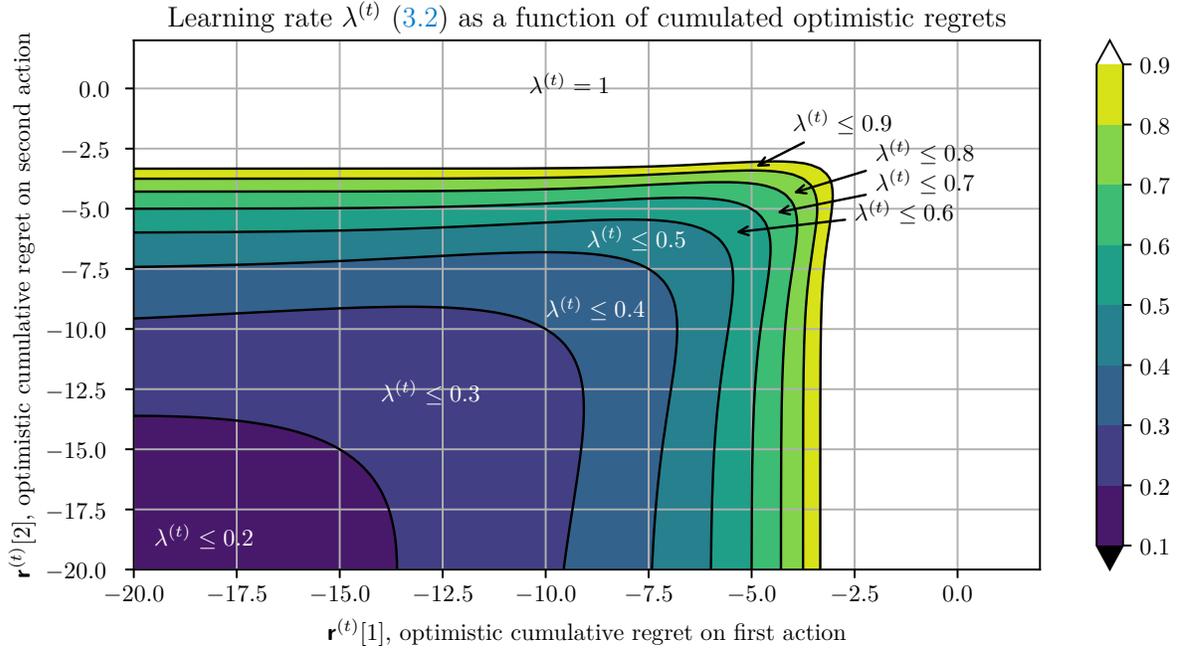}}
\caption{\textbf{Learning rate landscape:} Dependence of 
$\lambda\^t$, as defined in \eqref{eq:opt_problem_lambda}, on the optimistic regrets cumulated in $2$-action simplex. For the plot, the values $\eta=1, \alpha=4$ were chosen.}
\label{fig:lambda}
\end{figure}

\subsection{Properties of the Learning Rate Control Objective}

At first glance, it may not be immediately apparent that the maximization objective is tractable, or even concave in $\lambda$. However, the following result demonstrates that when $\alpha$ is chosen on the order of $\Omega(\log^2 d)$, $f(\lambda; \at)$ becomes self-concordant (and thus strongly concave) in $\lambda$. A concrete proof is provided in \Cref{sec:proof_temp_adjust}.

\begin{theorem}\label{lemma:convexity_concordance_step}
    For any $\at \in \bbR^d$, the rate control objective $f(\lambda; \at)$ defined in \eqref{eq:opt_problem_lambda} satisfies the following properties:
    \begin{itemize}
        \item Strong concavity: $f''(\lambda; \at) \leq -(\alpha - \log^2 d - 1)/\lambda^2$ for all $\lambda \in (0,\infty)$.
        \item Self-concordance: $(f'''(\lambda; \at))^2 \leq -4 f''(\lambda; \at)^{3}$,
    \end{itemize}
    where all derivatives are with respect to $\lambda$.
\end{theorem}

An immediate byproduct of \Cref{lemma:convexity_concordance_step} is that the learning rate control optimization problem is well-defined and admits a unique solution for a sufficiently large choice of the hyperparameter $\alpha$.

\begin{remark}\label{remark:unique_solution}
    The learning rate control objective $f(\lambda; \at)$ defined in \eqref{eq:opt_problem_lambda} admits a unique solution $\lambda\^t \in (0, \eta]$ for any given $\eta \geq 0$ and all $t \in \bbN$, whenever $\alpha$ is sufficiently large.
\end{remark}

A key technical property used in the analysis of the regret accumulated by \ours is the sensitivity of the solution to \eqref{eq:opt_problem_lambda} with respect to small perturbations in the regret vector $\at\^t$. This property is crucial, as it ensures that the dynamic learning rate $\lambda\^t$ for each player changes smoothly over time $t$, thereby allowing the self-play process to evolve in a stable and smooth manner.

\begin{theorem}(Sensitivity of learning rates on regrets) \label{theorem:stability_lambda}
    There exists a universal constant $\beta$,\footnote{For concrete values, choosing any $\beta \geq 70$ suffices.} such that for $\alpha \ge 2 + 2 \log d + \beta \log^2 d$, the following property holds.
    Let $\at, \aprimet \in \bbR^d$ be such that $\| \at - \aprimet \|_{\infty} \leq 2$, and let $\hat{\lambda}, \hat{\lambda'}$ be the corresponding learning rates, that is,
    \[
        \hat{\lambda} = \argmax_{t \in (0, \eta]} f(t; \at), \qquad
        \hat{\lambda'} = \argmax_{t \in (0,\eta]} f(t; \at').
    \]
    Then, $\hat{\lambda}$ and $\hat{\lambda'}$ are multiplicatively stable; specifically,
    \[
        \frac{7}{10} \leq  \frac{\hat{\lambda}}{\hat{\lambda}^\prime} \leq \frac{7}{5},
    \]
\end{theorem}

The proof of \cref{theorem:stability_lambda} is derived by combining an accurate analytical estimate, $\lambda_0$, for the value of $\widehat{\lambda}$ with techniques from the analysis of self-concordant functions. Specifically, by demonstrating that the intrinsic norm of the second-order ascent direction of $f$ at $\lambda_0$ is small, we conclude that the solution $\widehat{\lambda}$ lies within a small radius of $\lambda_0$ in the intrinsic norm. Moreover, using the bound on $f''(\lambda_0,\, \cdot\,)$ provided in \cref{lemma:convexity_concordance_step}, we establish the multiplicative proximity between $\lambda_0$ and $\widehat{\lambda}$. In particular, we show that the specific choice $\lambda_0 = (\alpha - 1)/(- \max_{r \in [d]} \{\at[r]\})$ serves as a reasonable analytical estimate, as the \emph{LogSumExp} function $\log \left( \sum_{\ind=1}^d e^{\lambda \at\^t[\ind]} \right)$ in the dynamic learning rate control problem~\ref{eq:opt_problem_lambda} behaves approximately like a softmax function over the regret vector $\at\^t$. Applying a similar procedure to the regret vector $\at^\prime$, we establish the multiplicative proximity between $\lambda_0^\prime$ and $\widehat{\lambda^\prime}$.

Combining these results, we conclude the multiplicative closeness of $\widehat{\lambda}$ and $\widehat{\lambda^\prime}$ in terms of $\at$, $\at^\prime$, and $\beta$. This process is formalized in the proof of the Multiplicative Stability Lemma~\ref{lemma:stability} in \Cref{sec:proof_temp_adjust}. Finally, by considering different cases for the values of $\max_{r \in [d]} \{\at[r]\}$ and $\max_{r \in [d]} \{\at^\prime[r]\}$ relative to $- \beta \log^2 d$, we establish the concrete multiplicative stability of $\widehat{\lambda}$ and $\widehat{\lambda^\prime}$, in light of \Cref{lemma:lambda_is_one}. Full details are provided in \Cref{sec:proof_temp_adjust}.

Since the analytic guess $\lambda_0$ can be computed efficiently and guarantees a small norm of the Newton step, the standard analysis of Newton's method for self-concordant functions immediately implies the following.
\begin{corollary} \label{coro:newton}
    Given any $\at \in \bbR^d$ and a desired \emph{relative} accuracy $\epsilon > 0$, $O(\log \log 1/\epsilon)$ iterations of Newton's method, starting from the initialization point $\lambda_0$, are sufficient to compute a point $\lambda$ that approximates $\lambda^* \defeq \argmin_{\lambda \in (0, \eta]} f(\lambda; \at)$ with relative error at most $\epsilon$, meaning that $(1 - \epsilon)\lambda^* < \lambda < (1 + \epsilon)\lambda^*$.
\end{corollary}

A byproduct of this analysis is that solving the learning rate control optimization problem~\ref{eq:opt_problem_lambda} up to a sufficient accuracy requires $O(d \log \log T)$ computational cost per iteration, which is negligible given that reading a reward vector already has complexity $O(d)$.

\begin{algorithm}[htp]
    \SetNoFillComment
    \caption{Kernelized Dynamic Learning Rate Control - Optimistic MWU (\kours)
    }\label{algo:ours_lambdaview_kernelized}
    \DontPrintSemicolon
    \KwData{Learning rate $\eta$, parameters $\alpha$ and $\beta$}\vspace{2mm}
    Set $ {\mut}\^1, \nut\^{0}, \vx\^{0} \gets \vec{0} \in \bbR^{d}, {\sigma}\^1, \lambda\^0 \gets 0 \in \bbR$\;
    \For{$t=1,2,\dots, T$}{
    Set $\mut\^t \gets {\mut}\^t + {\nut}\^{t-1}$ \Comment*{\color{commentcolor}Optimism for utility]\!\!\!\!\!}
    Set $\sigma\^t \gets {\sigma}\^t - \langle \nut\^{t-1}, \vx\^{t-1} \rangle$ \Comment*{\color{commentcolor}Optimism for correction]\!\!\!\!\!}
    
    \tcc{\color{commentcolor}\texttt{Dynamic Learning Rate Control via Kernelized Newton}}

    Set  $\displaystyle \lambda \leftarrow \lambda\^{t - 1} $ \Comment*{\color{commentcolor} Warm-start initialization for Newton]\!\!\!\!\!}
    
    \Repeat{\textup{Convergence of $\lambda$}}{
    \For{$r=1,2,\dots, d$}{
    Set $\displaystyle\vb[r] \leftarrow \exp\{\lambda \mut\^t[r]\}$ \Comment*{\color{commentcolor} See (\ref{eq:b_lambda})]\!\!\!\!\!} 
    }

    \For{$i=1,2,\dots, d$}{
    Set $\displaystyle \expect\mleft[\vv \mright]_i \leftarrow 1 - \frac{K_\Omega(\vb, \vebar_i)}{K_\Omega(\vb, \vec{1}_d)} $ \Comment*{\color{commentcolor} See (\ref{eq:moment1})]\!\!\!\!\!} 
    }

    \For{$i, j=1,2,\dots, d$}{
    \uIf{$ i \neq j$}{
    Set $\displaystyle \expect\mleft[\vv \vv^\top \mright]_{ij}  \leftarrow 1 + \frac{K_\Omega(\vb, \vebar_i)}{K_\Omega(\vb, \vec{1}_d)} + \frac{K_\Omega(\vb, \vebar_j)}{K_\Omega(\vb, \vec{1}_d)} - \frac{K_\Omega(\vb, \vebar_{ij})}{K_\Omega(\vb, \vec{1}_d)} $ \Comment*{\color{commentcolor} See (\ref{eq:moment2})]\!\!\!\!\!}
    }
    
    \Else{
    Set $\displaystyle \expect\mleft[\vv \vv^\top \mright]_{ii}  \leftarrow 1 + \frac{K_\Omega(\vb, \vebar_i)}{K_\Omega(\vb, \vec{1}_d)}  $ \Comment*{\color{commentcolor} See (\ref{eq:moment2.2})]\!\!\!\!\!}
    }
     
    }

    Set $\displaystyle f^\prime(\lambda; \At) \leftarrow (\mut\^t)^\top \expect\mleft[\vv \mright] + \sigma^t + \frac{\alpha - 1}{\lambda}$

    Set $\displaystyle f^{\prime\prime}(\lambda; \At) \leftarrow (\mut\^t)^\top \expect\mleft[\vv \vv^\top \mright] \mut\^t - ((\mut\^t)^\top \expect\mleft[\vv\mright])^2 - \frac{\alpha - 1}{\lambda^2}$

    \tcc{\color{commentcolor}\texttt{Newton Update}}

    Set $\displaystyle \lambda \leftarrow \lambda + \frac{f^\prime(\lambda; \At)}{f^{\prime\prime}(\lambda; \At)}$
    
    }

    Set  $\displaystyle \lambda\^{t} \leftarrow \lambda $ \medskip  
    
    \tcc{\color{commentcolor}\texttt{Kernelized \Opthedge}}
    \For{$r=1,2,\dots, d$}{
    $\displaystyle\vb\^t[r] \leftarrow \exp\{\lambda\^t \mut\^t[r]\}$ \Comment*{\color{commentcolor} See (\ref{eq:b_t})]\!\!\!\!\!} 
    }
    \For{$r=1,2,\dots, d$}{
    $\displaystyle\vx\^t[r] \leftarrow 1 - \frac{K_\Omega(\vb\^t, \vebar_r)}{K_\Omega(\vb\^t, \vec{1}_d)} $ \Comment*{\color{commentcolor} See (\ref{eq:komw_action})]\!\!\!\!\!} 
    }
    Play strategy $\displaystyle\vx\^t$ and observe $\vec \nu \^t \in \bbR^d$\;

    Set $\displaystyle \mut\^t \gets {\mut}\^t + {\nut}\^{t} - {\nut}\^{t-1}$ \Comment*{\color{commentcolor}Empirical cumulated corrections]\!\!\!\!\!}
    Set $\displaystyle \sigma\^t \gets {\sigma}\^t - \langle \nut\^t, \vx\^t \rangle + \langle \nut\^{t-1}, \vx\^{t-1} \rangle $ \Comment*{\color{commentcolor}Empirical cumulated utilities]\!\!\!\!\!}
    }
\end{algorithm}

\section{Extension to 0/1-Polyhedral Games via Kernels} \label{sec:kernel}

Before delving into the analysis of \ours, we remark that---similarly to the \Opthedge algorithm---our algorithm can be applied efficiently to certain classes of polyhedral convex games with $0$/$1$-integral vertices.

Consider a convex $0/1$-polyhedral set $\Omega \subseteq \mathbb{R}^d$, that is, a polytope whose set of vertices $\mathcal{V}_\Omega$ is a subset of $\{0,1\}^d$.
The \Opthedge algorithm can be directly applied on $\Omega$ by keeping track of a distribution $\vchi\^t$ over the vertices and updating the distribution multiplicatively depending on the utility $\Nut[\vv] = \langle \nut, \vv \rangle$ scored by each vertex $\vv \in \mathcal{V}_\Omega$. While this process takes time proportional to the number of vertices $|\mathcal{V}_\Omega|$ if implemented naively, \citet{farina2022kernelized} prove that the process can sometimes be simulated in polynomial time per iteration even when $|\mathcal{V}_{\Omega}|$ is large. Specifically, they show that the update of $\vchi\^t$ and computation of the expectation $\sum_{\vv \in \mathcal{V}_\Omega} \vchi\^t[\vv]\vv$ can be carried out using only \( d + 1 \) evaluations of a \emph{$0/1$-polyhedral kernel}, which they demonstrate can be evaluated efficiently in extensive-form games and various other convex 0/1-polyhedral settings, including $m$-sets, unit cubes, and flows on directed acyclic graphs by building on a prior idea of \citet{takimoto2003path}.

Although the dynamics of \ours are quite similar to those of \Opthedge, extending \ours to a kernelized version (\kours) is not an immediate consequence of \citet{farina2022kernelized}. This extension requires additional considerations, particularly due to the need to solve the dynamic learning rate control problem in \Cref{eq:opt_problem_lambda} at each time step $t$. We will show that this optimization problem can also be addressed using the kernel trick, with novel modifications. Let us begin with definitions of the $0/1$-polyhedral feature mapping, the associated kernel, and the key results from \citet{farina2022kernelized}.

\begin{definition}[0/1-polyhedral feature map and kernel \citep{farina2022kernelized}]
Associated with a convex $0/1$-polyhedral set $\Omega \subseteq \mathbb{R}^d$, define the \emph{0/1-polyhedral feature map} $\phi_\Omega : \mathbb{R}^d \rightarrow \mathbb{R}^{\mathcal{V}_\Omega}$,
\[
\phi_\Omega(\vx)[\vv] \defeq \prod_{\ind: \vv[\ind] = 1} 
\vx[\ind] \quad \forall \, \vx \in \mathbb{R}^d, \, \vv \in \mathcal{V}_\Omega,
\]
and the corresponding \emph{0/1-polyhedral kernel} $K_\Omega : \mathbb{R}^d \times \mathbb{R}^d \rightarrow \mathbb{R}$,
\[
K_\Omega(\vx_1, \vx_2) \defeq \langle \phi_\Omega(\vx_1), \phi_\Omega(\vx_2) \rangle = \sum_{\vv \in \mathcal{V}_\Omega}\prod_{\ind: \vv[\ind] = 1} \vx_1[\ind] \vx_2[\ind] \quad \forall \, \vx_1, \vx_2 \in \mathbb{R}^d.
\]
\end{definition}

We also define auxiliary indicator vectors $\vebar_i, \vebar_{ij} \in \mathbb{R}^d$ for all $i, j \in [d]$, such that
\[
\vebar_i [k] \defeq \mathbb{I}_{k \neq i} = 
\begin{cases}
    0 & \text{if } k = i, \\
    1 & \text{if } k \neq i,
\end{cases} \quad \text{and} \quad \vebar_{ij} [k] \defeq \mathbb{I}_{k \neq i \land k \neq j} = 
\begin{cases}
    0 & \text{if } k = i \text{ or } k = j, \\
    1 & \text{if } k \neq i \text{ and } k \neq j.
\end{cases}.
\]
The embedding of these vectors is useful in kernel computations.

The following proposition ensures that Step~(\ref{eq:softmax}),
\[
    \vchi\^t[\vv] \defeq \frac{\exp\{\lambda\^t \At\^t[\vv]\}}{\sum_{\vv^\prime \in \mathcal{V}_{\Omega}} \exp\{\lambda\^t \At\^t[\vv']\}} \qquad \forall \vv \in \mathcal{V}_{\Omega},
    \numberthis{eq:softmax_vertices}
\]
can be simulated using only $d + 1$ kernel evaluations, assuming that the dynamic learning rate $\lambda\^t$ is available. The key idea is that by embedding a carefully constructed vector $\vb\^t$ into the feature mapping $\phi_\Omega$, the algorithm’s updates—represented as the distribution over the vertices $\vchi\^t \in \Delta(\mathcal{V}_{\Omega})$—and thus the actions $\vx\^t \in \Omega$, become computable via the kernel $K_\Omega$.

\begin{proposition}[Theorem 4.1 and 4.2 of \citep{farina2022kernelized}] \label{prop:kernel_paper}
For all time steps $t \in T$, 
let $\mut\^t \defeq \nut\^t + \sum_{\tau=1}^t  \nut\^\tau$ be the optimistic sum of the utility vectors $\nut\^t$, and define the vector $\vb\^t \in \mathbb{R}^d$ as
\[
    \vb\^t[k] \defeq \exp\{\lambda\^t \mut\^t[k]\} \quad \forall \, k \in [d]. \numberthis{eq:b_t}
\]
Then, the distributions $\vchi\^t$ are proportional to $\phi_\Omega(\vb\^t)$,
\[
    \vchi\^t = \frac{\phi_\Omega(\vb\^t)}{K_\Omega(\vb\^t, \vec{1}_d)},
\]
and the iterates $\vx\^t$ produced by \ours are computed as
\[
    \vx\^t = \sum_{\vv \in \mathcal{V}_{\Omega}} \vchi\^t [\vv] \vv = \mleft[1 - \frac{K_\Omega(\vb\^t, \vebar_1)}{K_\Omega(\vb\^t, \vec{1}_d)}, 1 - \frac{K_\Omega(\vb\^t, \vebar_2)}{K_\Omega(\vb\^t, \vec{1}_d)}, \cdots, 1 - \frac{K_\Omega(\vb\^t, \vebar_d)}{K_\Omega(\vb\^t, \vec{1}_d)} \mright]. \numberthis{eq:komw_action}
\]
\end{proposition}

It remains to show that, for each time step $t$, the dynamic learning rate $\lambda\^t$ can be determined via kernelization. By \Cref{coro:newton}, we need to verify that the Newton steps of the optimization problem in \Cref{eq:opt_problem_lambda} can be simulated by the kernel $K_\Omega$. The Newton algorithm requires the calculation of $f^\prime(\lambda; \At)$ and $f^{\prime\prime}(\lambda; \At)$,\footnote{To simplify notation, we omit the superscript $t$ from this point onward in this section whenever it is clear from the context.}
\[
 f^\prime(\lambda; \At) & = \frac{\sum_{\vv \in \mathcal{V}_{\Omega}} \At[\vv] e^{\lambda \At[\vv]}}{\sum_{\vv \in \mathcal{V}_{\Omega}} e^{\lambda \At[\vv]}} + \frac{\alpha - 1}{\lambda} \\
 f^{\prime\prime}(\lambda; \At) & = \frac{\sum_{\vv \in \mathcal{V}_{\Omega}} \At[\vv]^2 e^{\lambda \At[\vv]} }{\sum_{\vv \in \mathcal{V}_{\Omega}} e^{\lambda \At[\vv]}} - \mleft( \frac{\sum_{\vv \in \mathcal{V}_{\Omega}} \At[\vv] e^{\lambda \At[\vv]}}{\sum_{\vv \in \mathcal{V}_{\Omega}} e^{\lambda \At[\vv]}}\mright)^2 - \frac{\alpha - 1}{\lambda^2},
\]
where the vector $\At$ is potentially of exponential size. We recall that by \Cref{eq:softmax_vertices}, $\vchi$ can be seen as a discrete random variable. We revisit the representation of $f^\prime(\lambda; \At)$ and $f^{\prime\prime}(\lambda; \At)$. Interestingly, they can be rewritten as
\[
f^\prime(\lambda; \At) & = \expect\mleft[\At[\vv] \mright] + \frac{\alpha - 1}{\lambda} \\
f^{\prime\prime}(\lambda; \At) & = \expect\mleft[ \At[\vv]^2\mright] - \expect\mleft[\At[\vv]  \mright]^2 - \frac{\alpha - 1}{\lambda^2},
\]
where the expectations are taken with respect to the distribution $\vv \sim \vchi$. Consequently, it suffices to verify that the first and second moments of $\At$ with respect to the distribution $\vchi$ can be computed via the kernelization approach. Given that
$\At[\vv] = \langle \mut , \vv \rangle + \sigma$, where
\[
\mut \defeq \nut\^t + \sum_{\tau=1}^t  \nut\^\tau \qquad\text{and}\qquad \sigma \defeq - (\langle \nut\^t, \vx\^t \rangle + \sum_{\tau=1}^t \langle \nut\^\tau, \vx\^\tau \rangle),
\]
we infer that
\[
    \expect\mleft[\At[\vv] \mright] = \expect\mleft[\langle \mut, \vv \rangle \mright] + \sigma = \langle  \mut, \expect\mleft[\vv \mright] \rangle + \sigma. 
\]
And,
\[
\quad \expect\mleft[\At[\vv]^2 \mright] = \expect\mleft[(\langle \mut, \vv \rangle + \sigma) ^2 \mright] = \expect\mleft[\langle \mut, \vv \rangle^2 \mright] + 2 \sigma \expect\mleft[\langle \mut, \vv \rangle \mright] + \sigma^2 = \mut^\top \expect\mleft[\vv \vv^\top \mright] \mut + 2 \sigma \langle \mut, \expect\mleft[\vv \mright] \rangle + \sigma^2,
\]
by multilinearity. We can simplify the $f^{\prime\prime}(\lambda; \At)$ term further,
\[
f^{\prime\prime}(\lambda; \At) & = \mut^\top \expect\mleft[\vv \vv^\top \mright] \mut + 2 \sigma \mut^\top \expect\mleft[\vv \mright] + \sigma^2 - (\mut^\top \expect\mleft[\vv \mright] + \sigma)^2 - \frac{\alpha - 1}{\lambda^2} \\
& = \mut^\top \expect\mleft[\vv \vv^\top \mright] \mut - (\mut^\top \expect\mleft[\vv\mright])^2 - \frac{\alpha - 1}{\lambda^2}
\]
Hence, it is adequate to calculate $\expect\mleft[\vv \mright]$ and $\expect\mleft[\vv \vv^\top \mright]$, which we formalize in the following Proposition.

\begin{proposition} \label{prop:kernel_order_two}
    Define the vector $\vb \in \mathbb{R}^d$ as
    \[
        \vb[k] \defeq \exp\{\lambda \mut[k]\} \quad \forall \, k \in [d]. \numberthis{eq:b_lambda}
    \]
    Then,
    \[
    \expect\mleft[\vv \mright] = \mleft[1 - \frac{K_\Omega(\vb, \vebar_1)}{K_\Omega(\vb, \vec{1}_d)}, 1 - \frac{K_\Omega(\vb, \vebar_2)}{K_\Omega(\vb, \vec{1}_d)}, \cdots, 1 - \frac{K_\Omega(\vb, \vebar_d)}{K_\Omega(\vb, \vec{1}_d)} \mright], \numberthis{eq:moment1}
    \]
    And,
    \[
    \expect\mleft[\vv \vv^\top \mright]_{ij} = 1 + \frac{K_\Omega(\vb, \vebar_i)}{K_\Omega(\vb, \vec{1}_d)} + \frac{K_\Omega(\vb, \vebar_j)}{K_\Omega(\vb, \vec{1}_d)} - \frac{K_\Omega(\vb, \vebar_{ij})}{K_\Omega(\vb, \vec{1}_d)} \numberthis{eq:moment2},
    \]
    for all $i, j \in [d]$, where $i \neq j$, and
    \[
        \expect\mleft[\vv \vv^\top \mright]_{ii} = 1 + \frac{K_\Omega(\vb, \vebar_i)}{K_\Omega(\vb, \vec{1}_d)} \numberthis{eq:moment2.2},
    \]
    for all $i \in [d]$.
\end{proposition}
\begin{proof}
    We start by $\expect\mleft[\vv \mright]$. 
    \[
    \expect\mleft[\vv \mright] = \sum_{\vv \in \mathcal{V}_{\Omega}} \vchi [\vv] \vv = \vx,
    \]
    the rest follows similar to \Cref{prop:kernel_paper}. Next, we analyze $\expect\mleft[\vv \vv^\top \mright]$. First we show that, for every $i. j \in [d]$,
    \[
    \phi_\Omega (\vebar_i) [\vv] = \prod_{k: \vv[k] = 1} \vebar_i[k] = \prod_{k: \vv[k] = 1} \mathbb{I}_{k \neq i} =  \mathbb{I}_{k \notin \vv}
    \]
    and by the fact that $\phi_\Omega (\vec 1_d) = \vec 1 $,
    \[
    \phi_\Omega (\vec 1_d)[\vv] - \phi_\Omega (\vebar_i) [\vv] =  \mathbb{I}_{i \in \vv}.
    \]
    Similarly,
    \[
    \phi_\Omega (\vebar_{ij}) [\vv] = \prod_{k: \vv[k] = 1} \vebar_{ij}[k] = \prod_{k: \vv[k] = 1} \mathbb{I}_{k \neq i \text{ and } k \neq j} =  \mathbb{I}_{i, j \notin \vv}
    \]
    For every $i. j \in [d]$,
    \[
    \expect\mleft[\vv_i \vv_j \mright] & = \sum_{\vv \in \mathcal{V}_{\Omega}} \vchi [\vv] \vv_i \vv_j \\
    & = \sum_{\vv \in \mathcal{V}_{\Omega}} \vchi [\vv].\mathbb{I}_{i, j \in \vv} \\
    & = \sum_{\vv \in \mathcal{V}_{\Omega}} \vchi [\vv] . ( 1 - \mathbb{I}_{i \notin \vy} - \mathbb{I}_{j \notin \vy} + \mathbb{I}_{i, j \notin \vv}) \numberthis{eq:inc_exc_princ} \\
    & = 1 - \sum_{\vv \in \mathcal{V}_{\Omega}} \vchi [\vv] . ( - \phi_\Omega (\vebar_i) [\vv] - \phi_\Omega (\vebar_j) [\vv] + \phi_\Omega (\vebar_{ij}) [\vv] ) \\
    & = 1 + \frac{K_\Omega(\vb, \vebar_i)}{K_\Omega(\vb, \vec{1}_d)} + \frac{K_\Omega(\vb, \vebar_j)}{K_\Omega(\vb, \vec{1}_d)} - \frac{K_\Omega(\vb, \vebar_{ij})}{K_\Omega(\vb, \vec{1}_d)}
    \] 
    where line (\ref{eq:inc_exc_princ}) follows from the inclusion–exclusion principle, and in the last line, we applied \Cref{prop:kernel_paper} multiple times. The case with $i = j$ is quite similar.
\end{proof}

An immediate byproduct of \Cref{prop:kernel_order_two} is that $d^2 + 1$ evaluations of Kernel $K_\Omega$ are sufficient for each iteration of the Newton optimization algorithm for Problem~(\ref{eq:opt_problem_lambda}).

\begin{corollary}
    \kours requires $(d + 1) + (d^2 + 1) \, O(\log \log T)$ kernel evaluations at each time step $t$, where the first $d + 1$ evaluations are used for Step~(\ref{eq:softmax_vertices}), and the remaining $(d^2 + 1) \, O(\log \log T)$ evaluations are for computing the dynamic learning rate $\lambda\^t$ via the Newton algorithm. \kours achieves $O(n \log^2 |\mathcal{V}_\Omega| \log T)$ regret in the self-play setting, and $\widetilde{O}(\sqrt{T \log |\mathcal{V}_\Omega|})$ regret in adversarial settings.
\end{corollary}

We conclude this section with the final piece of the puzzle for \Cref{algo:ours_lambdaview_kernelized} (\kours). Since computing the initial guess $\lambda_0 = (\alpha - 1)/(- \max_{\vv \in \mathcal{V}_{\Omega}} \{\At\^t[\vv]\})$—the initialization point for the Newton algorithm—is not necessarily efficient, we warm-start the Newton algorithm at each iteration $t$ by using the previous solution $\lambda_0 = \lambda\^{t - 1}$ from time $t - 1$. It is straightforward to show that, due to the multiplicative stability property of the learning rates (\Cref{theorem:stability_lambda}), setting $\lambda_0 = \lambda\^{t - 1}$ provides a sufficiently accurate analytical estimate for initializing the Newton algorithm. We formalize this observation below.

\begin{observation}\label{observation:newton_kernel}
    Given the dynamic learning rate control optimization problem~\cref{eq:opt_problem_lambda} for the regret vector $\At\^t$ of the game over $\Delta(\mathcal{V}_\Omega)$,
    \[
        \lambda\^t = \argmax_{\lambda \in (0, \eta]} \left\{ f(\lambda; \At\^t) \defeq (\alpha - 1) \log \lambda + \log \left( \sum_{\vv \in \mathcal{V}_\Omega} e^{\lambda \At\^t[\vv]} \right) \right\},
    \]
    Newton's method, warm-started from $\lambda_0 = \lambda\^{t-1}$, converges in $O(\log \log 1/\epsilon)$ iterations to the optimal solution $\lambda\^t$ with a relative error of at most $\epsilon$.
\end{observation}

\begin{proof}
    By the analysis of Newton’s method for self-concordant functions~\citep{renegar2001mathematical}, it is sufficient to show that the size of the Newton step at initialization $\lambda_0 = \lambda\^{t - 1}$ is small in the local norm, i.e.,
    \[
        \| n(\lambda_0) \|_{f^{\prime\prime}(\lambda_0; \At\^t)} = \frac{f^{\prime}(\lambda_0; \At\^t)^2}{|f^{\prime\prime}(\lambda_0; \At\^t)|} \leq 1,
    \]
    where $n(\lambda_0)$ denotes the Newton step. By the multiplicative stability of $\lambda\^t$ and our choice of initialization $\lambda_0 = \lambda\^{t - 1}$—as implied by the sensitivity of the learning rate to the regret vector in \Cref{theorem:stability_lambda}, and following similar reasoning to the proof of \Cref{theorem:stability_lambda}—this bound on the Newton step size can be readily established. We omit the details here in the interest of space.
\end{proof}

\section{Analysis of the Dynamic Learning Rate Control}
\label{sec:regret_analysis}

In this section, we present our results and analysis of the regret for \ours. A cornerstone of this analysis is the following regret bound (\Cref{theorem:rvu}), which follows the style of the \emph{nonnegative RVU property}, originally introduced by \citet{farina2022near}. This property differs from the original RVU property discussed in \Cref{section:background}, and is stronger in the sense that the nonnegative RVU property directly implies the RVU property.

\begin{theorem}[Nonnegative RVU bound of \ours] \label{theorem:rvu}
    Consider the cumulative regret $\tildereg\^T$ accrued by \ours algorithm up to time $T$. Assuming that $\|  {\nut}\^{t} \|_\infty \leq 1$ is satisfied for all $t \in [T]$, it follows that for any time $T \in \bbN$ and any learning rate $\eta \leq \frac{1}{50}$ and $\beta$ high enough ($\beta \geq 70$),
    \[
        \max\{0, \reg\^T\} \leq 3 + \frac{\alpha \log T + \log d}{\eta} + 6 \eta \sum_{t=1}^{T-1} \| \nut\^{t} - \nut\^{t-1} \|_{\infty}^2 - \frac{1}{24\eta} \sum_{t=1}^{T-1} \| \vx\^{t+1} - \vx\^{t}  \|_{1}^2.
    \]
\end{theorem}

The proof sketch of our nonnegative RVU bound is provided in \Cref{sec:proof_sketch}. The use of positive regret, $\max\{0, \reg\^T\}$, is pivotal in our analysis, as any upper bound on the sum $\sum_{i=1}^{n} \max\{0, \reg\^T_i\}$ directly implies the same upper bound on the maximum regret, $\max_{i \in [n]} \reg\^T_i$, among the $n$ players due to nonnegativity. This, in turn, implies that the empirical distribution of joint strategies converges to an approximate CCE of the game.

With \Cref{theorem:rvu} at hand, the path forward becomes straightforward. The general plan is to use the nonnegativity of the sum $\sum_{i=1}^{n} \max\{0, \reg\^T_i\}$ and to set a sufficiently small learning rate $\eta$ in order to infer that the total path length of the play, i.e., $\sum_{i = 1}^{n} \sum_{t = 1}^{T - 1} \|\vx\^{t+1}_i - \vx\^{t}_i \|_1^2$, is bounded by $O(n \log^2 d \log T)$. This procedure is formalized in \Cref{theorem:bound_on_path_length}.

\begin{theorem}[Bound on total path length] \label{theorem:bound_on_path_length}
    Under \Cref{assumption:bound}, if all the players follow \ours algorithm with learning rate $\eta \leq \min\{\frac{1}{50}, \frac{1}{12 \sqrt{2} L n} \}$, then
    \[
    \sum_{i = 1}^{n} \sum_{t = 1}^{T - 1} \|\vx\^{t+1}_i - \vx\^{t}_i \|_1^2 \leq 144 n \eta + 48 n (\alpha \log T + \log d).
    \]
\end{theorem}

Given the RVU bound of \Cref{theorem:rvu} and the bound on the total path length from \Cref{theorem:bound_on_path_length}, we can adhere to the standard machinery of RVU bounds~\citep{Syrgkanis15:Fast} and prove that the regret $\reg_i\^T$ of each player $i \in [n]$ is bounded by $O(n \log^2 d \log T)$, as stated in \Cref{theorem:regret_bound}. The idea is that, under \Cref{assumption:bound}, the variation in the utilities observed over time, $\sum_{t=1}^{T-1} \| \nut\^{t} - \nut\^{t-1} \|_{\infty}^2$, can be bounded by the total path length, thereby yielding the result.

\begin{theorem}[Regret bound of \ours (Formal version of \Cref{theorem:informal_regret_bound})] \label{theorem:regret_bound}
    Under \Cref{assumption:bound}, if all the players $i \in [n]$ follow \ours with a learning rate $\eta = \min\left\{\frac{1}{50}, \frac{1}{12 \sqrt{2} L n} \right\}$, then the regret $\reg_i\^T$ of each player $i \in [n]$ is bounded by
    \[
        \reg_i\^T \leq 6 + \max\left\{50 + 12 \sqrt{2} L n, 24 \sqrt{2} L n\right\} (\alpha \log T + \log d) = O(n \log^2 d \log T).
    \]
    Additionally, the \ours algorithm for each player $i \in [n]$ is adaptive to adversarial utilities, meaning that the regret incurred in the face of adversarial utilities is $\reg_i\^T = \Tilde{O}(\sqrt{T \log d})$.
\end{theorem}

For the detailed proofs, please refer to \Cref{sec:main_results_proofs}. We conclude this section by noting that our regret guarantees in \Cref{theorem:regret_bound} hold \emph{uniformly for all $T \in \mathbb{N}^+$ simultaneously}, since the algorithm \ours and the choice of learning rate $\eta$ do not depend on the time horizon $T$. As a result, even when the horizon of self-play $T$ is unknown, or miscoordination occurs among the participants, \ours still enjoys low regret guarantees in both self-play and adversarial settings—without the need for additional techniques such as the doubling trick.

This is in contrast to the analysis of \citet{daskalakis2021near}—and indeed our \emph{anytime convergence} provides a strictly stronger guarantee—where the choice of learning rate $\eta$ depends on the time horizon $T$. Consequently, when $T$ is unknown, additional techniques such as the doubling trick appear to be unavoidable. Furthermore, since the learning rate $\eta$ in \citet{daskalakis2021near} is restricted to the interval $1/T \leq \eta \leq 1/(C n \log^4 T)$ for some constant $C$ (see Lemmas 4.2 and C.4 in~\citep{daskalakis2021near}), their guarantees fail to apply in the regime where $T \leq C n \log^4 T$, which frequently arises in games with many players or in short-horizon scenarios.

\subsection{Design of the \ours Algorithm and Equivalent Viewpoints}
\label{sec:design_dynamic_learning}

Before delving into the proof sketch in \Cref{sec:proof_sketch}, we first discuss how to mathematically formalize the concept of dynamic learning rate control, leading to \Cref{algo:ours_lambdaview} (\ours). In this context, we also present alternative perspectives on \ours and highlight how they contribute to both its regret analysis and computational properties.

We begin with the standard dynamics of Optimistic Follow-the-Regularized-Leader (\OFTRL) algorithms with a potentially dynamic learning rate $\lambda\^t$,
\[
\vx\^t \gets \argmax_{\vx \in \mathcal{X}} \left\{ \lambda\^t \langle \at\^t, \vx \rangle + \phi(\vx) \right\}, \numberthis{eq:FTRL_naive}
\]
where $\lambda\^t \in (0, \eta]$ is chosen according to a separate dynamic that we will design later, and $\phi$ is a regularizer over the space $\mathcal{X}$. In the case of the negative entropy regularizer, $\phi(\vx) = \sum_{\ind = 1}^{d} \vx[\ind] \log \vx[\ind]$, Formulation~(\ref{eq:FTRL_naive}) recovers the celebrated Optimistic MWU algorithm with learning rate $\lambda\^t$.

The fundamental idea here is to also integrate the dynamic change of $\lambda\^t$ into \Cref{eq:FTRL_naive}. Naturally, in online learning problems, we aim to incentivize selecting actions with the best rewards. On the other hand, for the purpose of self-play, as discussed comprehensively in \Cref{sec:loga}, we seek to additionally \emph{pace down} the learner when it is \emph{performing too well}. Thus, a seamless extension of Optimization Step~\ref{eq:FTRL_naive}, equipped with an automatic dynamic adjustment of $\lambda\^t$, takes the form:
\[
\vstack{\lambda\^t}{\vx\^t} \gets \argmax_{\lambda \in (0, \eta], \vx \in \mathcal{X}} \left\{ \lambda \langle \at\^t, \vx \rangle + \phi(\vx) \right\}. \numberthis{eq:FTRL_naive1}
\]

However, in this formulation, $\lambda\^t$ exhibits behavior akin to that of a step function. When $\langle \at\^t, \vx\^t \rangle > 0$, this formulation naively reduces to $\lambda\^t = \eta$, which corresponds to the original \OFTRL with a constant learning rate; otherwise, it trivially sets $\lambda\^t = 0$ and $\vx\^t =  \argmax_{\vx \in \mathcal{X}} \phi(\vx)$. Instead, to improve the predictability of the behavior of the dynamics $\vx\^t$ in the self-play setting, we need a smoother transition for the dynamic learning rate $\lambda\^t$. To achieve this, we can incorporate an additional regularizer $\rho(\lambda)$ for $\lambda \in (0, \eta]$ that defines a smooth thresholding procedure for what it means to ``\emph{perform too well}'' relative to the regret vector $\at\^t$ into the dynamics of \Cref{eq:FTRL_naive1}, which leads to
\[
\vstack{\lambda\^t}{\vx\^t} \gets \argmax_{\lambda \in (0, \eta], \vx \in \mathcal{X}} \left\{ \lambda \langle \at\^t, \vx \rangle + \rho(\lambda) + \phi(\vx) \right\}. \numberthis{eq:DRLC_FTRL}
\]
In other words, the term $\lambda \langle \at\^t, \vx \rangle$ is motivating \emph{regret minimization} for the player, while the regularizer $\rho(\lambda)$ is prohibiting the player from having a \emph{very low regret}, thereby creating a \emph{balance among the performance of the players} of the game.

We note that the dynamics of Formulation~(\ref{eq:DRLC_FTRL}) are not jointly concave in $(\lambda, \vx)$ for a general choice of $\rho$ and $\phi$; therefore, its \emph{computational} aspects and \emph{regret analysis} remain unclear. Notably, for the special choice of $\rho(\lambda) \defeq (\alpha - 1) \log \lambda$, $\phi(\vx) = \sum_{\ind = 1}^{d} \vx[\ind] \log \vx[\ind]$, and $\mathcal{X} = \Delta^d$, we will demonstrate that Formulation~(\ref{eq:DRLC_FTRL}) is equivalent to \ours. Hence, it can be computed at each iteration by solving the 1-dimensional optimization problem of dynamic learning rate control (\ref{eq:opt_problem_lambda}) and the play of \Opthedge, as outlined in \Cref{algo:ours_lambdaview} and \Cref{sec:learning_rate_control_problem}.

To analyze the regret of \ours, we show that \Cref{algo:ours_lambdaview} can be reformulated as a specific instance of \OFTRL, beyond its original form in Formulation (\ref{eq:DRLC_FTRL}). A key technical step in this analysis is to express the iterates $\vx\^t$ produced by the algorithm through the \OFTRL perspective. By defining $\vy\^t := \lambda\^t \vx\^t \in (0,1]\Delta^d = \mleft\{ \vy \in \mathbb{R}_{\geq 0}^d \: \big| \; \sum_{\ind = 1}^d \vy[\ind] \leq 1 \mright\}$, we will demonstrate that the iterates $\vy\^t$ satisfy the \OFTRL rule,
\[
    \vy\^{t} \defeq \argmax_{\vy \in (0,1]\Delta^d} \left\{\langle \at\^t, \vy\rangle + \psi(\vy) \right\}, \numberthis{eq:oftrl_psi}
\]
where $\psi$ is a special regularizer with strong spectral properties (to be discussed in \Cref{sec:proof_sketch}),
\[
    \psi(\vy) = - \frac{1}{\eta} \alpha \log \left(\sum_{\ind=1}^d \vy[\ind]\right) + \frac{1}{\eta}\frac{1}{\sum_{\ind=1}^d \vy[\ind]} \sum_{\ind=1}^d \vy[\ind] \log \vy[\ind]. \numberthis{eq:regularizer}
\]
Interestingly, this formulation looks similar to the lifting idea of \citet{farina2022near}, where instead of applying \OFTRL on the original space $\mathcal{X} = \Delta^n$, \OFTRL is designed over the lifted space $\widetilde{\mathcal{X}} = \{(\gamma, \vy) : \gamma \in [0, 1], \vy \in \gamma \mathcal{X}\} = [0, 1] \Delta^n$, but without the need to incorporate the lifting parameter $\gamma$ into the dynamics of the algorithm (neither in the utility term nor in the regularizer). We note that despite this similar appearance, as will be seen in the next sections, the analysis of \ours is substantially different from that of Log-Regularized Lifted \OFTRL (\LRLOFTRL) \citep{farina2022near}, where the log regularizer is chosen over the lifted space $\widetilde{\mathcal{X}}$, i.e., $\psi(\gamma, \vy) = - \frac{1}{\eta} \log \gamma -  \frac{1}{\eta} \sum_{\ind = 1}^{d} \log \vy[\ind]$, to ensure elementwise multiplicative stability of actions due to the high curvature of the log regularizers and the structure of the induced intrinsic norms.

The proof sketch of the regret analysis using Formulation~(\ref{eq:oftrl_psi}) is discussed in \Cref{sec:proof_sketch}. For the moment, we formalize these three alternative formulations of \ours in the following theorem and defer the equivalence proofs to \Cref{sec:equivalent_viewpoints}.

\begin{theorem} \label{theorem:equivalence_viewpoints}
    \Cref{algo:ours_lambdaview} (\ours) can alternatively be viewed as \oursgeneral (Formulation \ref{eq:DRLC_FTRL}) with the choice of regularizers $\rho(\lambda) = (\alpha - 1) \log \lambda$ and $\phi(\vx) = \sum_{\ind=1}^d \vx[\ind] \log \vx[\ind]$, i.e.,
    \[
        \vstack{\lambda\^t}{\vx\^t} \gets \argmax_{\lambda \in (0, \eta], \vx \in \Delta^d}\mleft\{ {\lambda} \mleft\langle  \at\^t, \vx \mright\rangle + (\alpha - 1) \log \lambda - \sum_{\ind=1}^d \vx[\ind] \log \vx[\ind]\mright\} \numberthis{eq:drlc_ftrl_special}.
    \]
    or additionally as \OFTRL with the regularizer $\psi$ defined in \Cref{eq:regularizer}, i.e.,
    \[
        \vy\^t \gets \argmax_{\vy \in (0,1]\Delta^d }\mleft\{ {\eta} \mleft\langle  \at\^t, \vy \mright\rangle + \alpha \log \left(\sum_{\ind=1}^d \vy[\ind]\right) - \frac{1}{\sum_{\ind=1}^d \vy[\ind]} \sum_{\ind=1}^d \vy[\ind] \log \vy[\ind]\mright\},
    \]
    In other words, the three perspectives (\Cref{algo:ours_lambdaview}, Formulations \ref{eq:drlc_ftrl_special} and \ref{eq:oftrl_psi}) of \ours are equivalent and result in the same learning dynamics.
\end{theorem}

\subsection{Proof Sketch} \label{sec:proof_sketch}

As discussed in the previous section, an essential step in the regret analysis of \ours is that, by \Cref{theorem:equivalence_viewpoints}, the play $\vx\^t$ generated by \ours is equivalent to $\vy\^t$ with a transformation by the dynamic learning rate $\lambda\^t$, where the iterates of $\vy\^t$ follow
\[
    \vy\^{t} \defeq \argmax_{\vy \in (0,1]\Delta^d} \left\{\langle \at\^t, \vy\rangle + \psi(\vy) \right\}, 
\]
with $\psi$ being our regularizer $\psi(\vy) :\Omega  \rightarrow \bbR$ on the domain 
\[
\Omega \defeq [0,1]\Delta^d = \{\vy \in \bbR_+^d \mid \sum_{\ind = 1}^{d} \vy[\ind] \leq 1\} 
\]
defined as
\[
    \psi(\vy) = - \frac{1}{\eta} \alpha \log \left(\sum_{\ind=1}^d \vy[\ind]\right) + \frac{1}{\eta}\frac{1}{\sum_{\ind=1}^d \vy[\ind]} \sum_{\ind=1}^d \vy[\ind] \log \vy[\ind],
\]
where $\alpha \defeq 2 + 2 \log d + \beta \log^2 d$ with $\beta \geq 70$ is a hyperparameter. We restate that the dynamic learning rate parameter $\lambda\^t$ is the unique solution to the Dynamic Learning Rate Control Problem~\ref{eq:opt_problem_lambda} for each time $t \in [T]$.

\paragraph{Notation.} For the analysis of the algorithm, we introduce additional necessary notation for convenience. We use $\| {\vz} \|_{\vec{B}}^2 \defeq {\vz}^\top \vec{B} {\vz}$ to denote the quadratic norm induced by the matrix $\vec{B} \in \mathbb{R}^{d \times d}$. Let $\summ{\vy} = \sum_{\ind = 1}^{d} \vy[\ind]$ be the sum of the elements of the vector $\vy  \in \mathbb{R}^d$. For two vectors $\vy, \vz \in (0, 1]\Delta^d$ in $(0, 1]\Delta^d$ space, define $\vx, \vtheta \in \Delta^d$ as the induced actions in the simplex, with $\vx[r] = \frac{\vy[r]}{\summ{\vy}}$ and $\vtheta[r] = \frac{\vz[r]}{\summ{\vz}}$ for every coordinate $r \in [d]$.

\paragraph{Spectral Properties of $\psi$.}
Prior to discussing the analysis of regret and deriving the nonnegative RVU bounds in \Cref{theorem:rvu}, we examine the strong spectral properties of the regularizer $\psi$, which may be of independent interest for future studies. Proofs of these properties are provided in \Cref{sec:spectral_psi}. In \Cref{theorem:regularizer_convex}, we prove that the regularizer $\psi$ is not only strongly convex, but its curvature is lower bounded, interestingly, by the diagonal matrix $\frac{1}{2 \eta \sum_{\ind = 1}^d \vy[\ind]} \operatorname{diag}\left(\frac{1}{\vy[1]}, \frac{1}{\vy[2]}, \ldots, \frac{1}{\vy[d]}\right)$. It is notable that these strong spectral properties of the regularizer $\psi$ hold only when $\alpha \geq 2 + 2 \log d + 2 \log^2 d$, and it is not possible to eliminate the dependence of the parameter $\alpha$ on $\log^2 d$, which makes this regime particularly interesting. 

\begin{theorem} \label{theorem:regularizer_convex}
    For $\alpha \ge 2 + 2 \log d + 2 \log^2 d$, the function $\psi (\vy)$ is strongly convex, and its positive definite Hessian at any point satisfies the bound
    \begin{align*}
        \nabla^2 \psi(\vy) \succeq \frac{1}{2 \eta} \begin{pmatrix}
            \frac{1}{\vy[1] \cdot \sum_{\ind=1}^d \vy[\ind]} & 0                                                & \cdots & 0                                                \\
            0                                                & \frac{1}{\vy[2] \cdot \sum_{\ind=1}^d \vy[\ind]} & \cdots & 0                                                \\
            \vdots                                           & \vdots                                           & \ddots & \vdots                                           \\
            0                                                & 0                                                & \cdots & \frac{1}{\vy[d] \cdot \sum_{\ind=1}^d \vy[\ind]}
        \end{pmatrix}.
    \end{align*}
\end{theorem}

Consequently, by strong convexity of the regularizer $\psi$ (see \cref{theorem:regularizer_convex}, we can define its induced Bregman divergence $D_\psi$,
\[
    D_\psi (\vy \,\big\|\, \vz) & = \psi(\vy) - \psi(\vz) - \langle \nabla\psi(\vz), \vy - \vz \rangle.
\]
In turn, we establish the following spectral properties of the induced Bregman divergence $D_\psi$, which are pivotal for the regret analysis to be discussed next. 
\begin{theorem} \label{theorem:bregman_curve}
    The Bregman divergence $D_\psi(\cdot \,\big\|\, \cdot)$ induced by the regularizer $\psi(\cdot)$ satisfies the following properties:
    \begin{itemize}
        \item \textbf{Curvature in the lifted space $(0, 1] \Delta^d$}: $D_\psi(\cdot \,\big\|\, \cdot)$ is lower bounded by a term proportional to the $\ell_1$ norm on the lifted simplex $(0, 1] \Delta^d$:
        \[
        D_\psi(\vy \,\big\|\, \vz) \geq \frac{1}{2 \eta} \| \vy - \vz \|_1^2.
        \]
        
        \item \textbf{Curvature in the action simplex $\Delta^d$}: $D_\psi(\cdot \,\big\|\, \cdot)$ is lower bounded by a term proportional to the $\ell_1$ norm on the action simplex $\Delta^d$:
        \[
            D_{\psi}(\vz \,\big\|\, \vy) \geq \frac{1}{4 \eta} (1 - \epsilon) \|\vtheta - \vx \|_1^2,
        \]
        under the multiplicative stability assumption over the sum of actions, $\omega \defeq \frac{\summ{\vz}}{\summ{\vy}} \in [1 - \epsilon, 1 + \epsilon]$ for a constant $\epsilon \in (0, \frac{2}{5})$.
    \end{itemize}
\end{theorem}

This theorem is proved in \Cref{prop:el1_lifted,prop:el1_simplex} in \Cref{sec:spectral_psi}. The proof of \Cref{prop:el1_lifted} follows from observing that \Cref{theorem:regularizer_convex} implies that the regularizer $\psi$ is strongly convex with respect to the $\ell_1$ norm. The proof of \Cref{prop:el1_simplex} is more involved and relies on the special representation of $D_\psi(\cdot \,\big\|\, \cdot)$ shown in \Cref{prop:bregman_rep}, accompanied by the multiplicative stability property, finite differences of entropies, and Pinsker's inequality.

\paragraph{Nonnegative Regret.} Recall that the dynamics of \ours are equivalent to those of \Cref{eq:oftrl_psi}. Thus, we begin by analyzing the regret of the \OFTRL in \Cref{eq:oftrl_psi}. We denote this regret by $\tildereg\^T$,
\[
\tildereg\^T \defeq \max_{\vy^* \in [0,1]\Delta^d } \sum_{t = 1}^{T} \langle  {\ut}\^{t}, \vy^* - \vy\^{t} \rangle.
\]
Recall that our nonnegative RVU bounds in \Cref{theorem:rvu} are on the positive regret $\max\{0, \reg\^T\}$. By the following proposition, we connect the $\tildereg\^T$ to $\reg\^T$ defined in \Cref{eq:linreg}.

\begin{proposition} \label{prop:reg+}
    For any time horizon $T \in \mathbb{N}$, we have that $\tildereg\^T = \max\{0, \reg\^T\}$. As a result, $\tildereg\^T \geq 0$ and $\tildereg\^T \geq \reg\^T$.
\end{proposition}

This proposition is proved in \Cref{sec:pos_regret}. An immediate consequence of this proposition is that any RVU bounds for the \OFTRL in \Cref{eq:oftrl_psi} immediately imply nonnegative RVU bounds for \ours. Hence, we proceed by analyzing $\tildereg\^T$.

\paragraph{Nonnegative RVU.} With the equivalent formulation of \ours as the \OFTRL algorithm in \Cref{eq:oftrl_psi}, the strong spectral properties of $\psi$, and the nonnegative Regret, the analysis follows the standard machinery of Optimistic Follow-the-Regularized-Leader algorithms.

Consequently, we define $ {\vy}\^{t}$ as the outputs generated by the \OFTRL in \Cref{eq:oftrl_psi}:
\begin{align*}
    {\vy}\^{t} = \argmax_{  {\vy} \in \Omega} - F_t(  {\vy}) = \argmin_{  {\vy} \in \Omega} F_t(  {\vy}), \textup{ where } F_t(  {\vy}) = - \mleft\langle  {\Ut}\^t +  {\vec\ut}\^{t-1},  {\vy} \mright\rangle + \psi(  {\vy}). \numberthis{eq:Fdef}
\end{align*}
Moreover, we define the auxiliary sequence $ {\vz}\^{t}$ as the solutions to the corresponding standard \FTRL for each time $t$:
\begin{align*}
    {\vz}\^{t} = \argmax_{  {\vz} \in \Omega} - G_t(  {\vz}) = \argmin_{  {\vz} \in \Omega} G_t(  {\vz}), \textup{ where } G_t(  {\vz}) = - \mleft\langle  {\Ut}\^t,  {\vz} \mright\rangle + \psi(  {\vz}). \numberthis{eq:Gdef}
\end{align*}
Functions $G_t$ and $F_t$ are strongly convex, as shown by \Cref{theorem:regularizer_convex}. We continue the analysis with the following standard lemma from \OFTRL analysis.

\begin{lemma} \label{lemma:basi_OFTRL}
    For any $ {\vy} \in \Omega$, and the sequence $\{ {\vy}\^{t}\}_{t = 1}^{T}$ generated by \Cref{eq:Fdef}, it holds that
    \[
        \sum_{t = 1}^{T} \mleft\langle  {\vy} -  {\vy}\^{t},  {\ut}\^{t} \mright\rangle \leq \psi( {\vy}) - \psi( {\vy}\^{1}) + \underbrace{\sum_{t = 1}^{T} \mleft\langle  {\vz}\^{t+1} -  {\vy}\^{t},  {\ut}\^{t} -  {\ut}\^{t - 1}\mright\rangle}_{\textup{(I)}} \\
        \qquad \underbrace{- \sum_{t = 1}^{T} \mleft( D_{\psi}( {\vy}\^{t} \,\big\|\,  {\vz}\^{t}) + D_{\psi}( {\vz}\^{t+1} \,\big\|\,  {\vy}\^{t}) \mright)}_{\textup{(II)}}.
    \]
\end{lemma}

The general plan for our RVU bounds is to carefully upper bound term (I) to obtain the beta terms $O(\eta) \sum_{t=1}^{T-1} \| \nut\^{t} - \nut\^{t-1} \|_{\infty}^2$, and to upper bound term (II) to obtain the gamma terms $- \Omega\left(\frac{1}{\eta}\right) \sum_{t=1}^{T-1} \| \vx\^{t+1} - \vx\^{t} \|_{1}^2$. This procedure is formalized in the proof of \Cref{theorem:rvu} in \Cref{sec:proof_rvu}.

A key observation in this regard is that, by \Cref{theorem:bregman_curve} and the multiplicative stability of the Dynamic Learning Rate Control Step in \Cref{theorem:stability_lambda}, the Bregman divergence $D_\psi(\vy \,\big\|\, \vz)$ is not only lower bounded by the squared norm of the difference of actions in the lifted space, $\Omega\left(\frac{1}{\eta}\right) \| \vy - \vz \|_1^2$ (\Cref{lemma:bregman_to_l1}), but also, due to our novel choice of regularizer, it is lower bounded by the squared norm of the difference of actions in the original space (action simplex), $\Omega\left(\frac{1}{\eta}\right) \| \vtheta - \vx \|_1^2$ (\Cref{lemma:beta_terms}). Combining these two bounds, we obtain that 
\[
    D_\psi(\vy \,\big\|\, \vz) \geq \Omega \left(\frac{1}{\eta}\right) \left(\| \vy - \vz \|_1^2 + \| \vtheta - \vx \|_1^2\right). \numberthis{eq:lower_breg_order}
\]
By choosing a small learning rate $\eta$, the $\Omega \left(\frac{1}{\eta}\right) \| \vy - \vz \|_1^2$ terms in \Cref{eq:lower_breg_order} compensate for the $O(\eta) \| \vy - \vz \|_1^2$ that results from converting term (I) into the beta terms $O(\eta) \sum_{t=1}^{T-1} \| \nut\^{t} - \nut\^{t-1} \|_{\infty}^2$. This conversion process follows from the relationship between the $\ut$s and $\nut$s, basic calculations involving the Cauchy–Schwarz and Hölder’s inequalities, and a sufficiently small choice of $\eta$. Details are provided in the proof of \Cref{theorem:rvu} in \Cref{sec:proof_rvu}.

\subsection{Detailed Analysis}
In this section, we first state the auxiliary lemmas used in the analysis, and then provide the detailed analysis and formal proofs for the sketches presented in the previous sections.

\subsubsection{Auxiliary Lemmas}
\begin{theorem}[Theorem 2.2.5 of \citep{renegar2001mathematical}] \label{theorem:newton_self}
    Given a self-concordant function $f: \bbR^d \rightarrow \bbR$, and the induced local norm $\| . \|_{\vx} \defeq \| . \|_{\nabla^2 f(\vx)}$, if for some $\vx$ in the domain of $f$ we have $\| n(\vx) \|_{\vx} \leq \frac{1}{4}$, then $f$ has a minimizer $\vz$ that
    \[
        \|z - x \|_x \leq 3 \| n(\vx) \|_{\vx},
    \]
    where $n(\vx)$ is the Newton update, $n(\vx) = - [\nabla^2 f(\vx)]^{-1} \nabla f(\vx)$. 
    By 
\end{theorem}

\begin{lemma} \label{lemma:softmax_bound}
    For any set of numbers $s_1, s_2, \dots, s_d$, the log-sum-exp function $g(t) \defeq h(t s_1, t s_2, \dots, t s_d) = \log \mleft( \sum_{\ind = 1}^{d} \exp\{t s_\ind\}\mright) $ satisfies,
    \[
        \max\{ s_1, s_2, \dots, s_d \} \leq \frac{g(t)}{t} \leq \max\{ s_1, s_2, \dots, s_d \} + \frac{\log d}{t}.
    \]
\end{lemma}
\begin{proof}
    The LHS follows by lowerbounding $\sum_{\ind = 1}^{d} \exp\{t o_\ind\}$ by $\max\{\exp\{t o_\ind\}\}$. The RHS follows by upperbounding $\sum_{\ind = 1}^{d} \exp\{t o_\ind\}$ by $d.\max\{\exp\{t o_\ind\}\}$.
\end{proof}

\begin{lemma}[Pinsker's inequality] \label{lemma:pinsker}
    Given two discrete random variables ${\vec p}$ and ${\vec q}$,
    \[
    \| {\vec p} - {\vec q} \|_1^2 \leq 2 \kl{\vec p}{\vec q}.
    \]
\end{lemma}

\begin{lemma}[Entropy difference] \label{lemma:entropy_diff}
    Given two discrete random variables ${\vec p}$ and ${\vec q}$ with support size $d$,
    \[
    | \ent{\vec p} - \ent{\vec q} | \leq (\log d)\sqrt{2 \kl{\vec p}{\vec q}}.
    \]
\end{lemma}

\subsubsection{Proofs for Learning Rate Control Step (\Cref{sec:learning_rate_control_problem})} \label{sec:proof_temp_adjust}
We start proving \Cref{lemma:convexity_concordance_step}, by first proving strong concavity of optimization problem \eqref{eq:opt_problem_lambda}.
\begin{lemma}(Strong concavity of learning rate control Step) \label{lemma:convexity_step}
    The objective of the optimization problem,
    \[
        \hat{\lambda} \defeq \argmax_{\lambda \in (0, \eta]} \left\{f(\lambda; \at) \defeq (\alpha - 1) \log \lambda + \log \Big( {\sum_{\ind=1}^d e^{\lambda \at[\ind]}} \Big)  \right\},
    \]
    is $\frac{\alpha - \log^2 d - 1}{\lambda^2}$-strongly concave.
\end{lemma}
\begin{proof}
    Taking derivatives from $f(\lambda; \at)$ w.r.t. $\lambda$,
    \[
        \frac{\partial f(\lambda; \at)}{\partial \lambda}     & = \frac{\sum_{\ind = 1}^d \at[\ind] e^{\lambda \at[\ind]}}{\sum_{\ind = 1}^d e^{\lambda \at[\ind]}} + \frac{\alpha - 1}{\lambda}.                                                                                                                                                        \\
        \frac{\partial^2 f(\lambda; \at)}{\partial \lambda^2} & = \frac{ \Big( \sum_{\ind = 1}^d \at[\ind]^2 e^{\lambda \at\^t[\ind]} \Big) \Big(\sum_{\ind = 1}^d e^{\lambda \at\^t[\ind]} \Big) - \Big(\sum_{\ind = 1}^d \at[\ind] e^{\lambda \at[\ind]} \Big)^2}{\Big(\sum_{\ind = 1}^d e^{\lambda \at[\ind]} \Big)^2} - \frac{\alpha - 1}{\lambda^2} \\
        & = \frac{\sum_{\ind = 1}^d \at[\ind]^2 e^{\lambda \at[\ind]} }{\sum_{\ind = 1}^d e^{\lambda \at[\ind]}} - \mleft( \frac{\sum_{\ind = 1}^d \at[\ind] e^{\lambda \at[\ind]}}{\sum_{\ind = 1}^d e^{\lambda \at[\ind]}}\mright)^2 - \frac{\alpha - 1}{\lambda^2}.
    \]
    Substituting $\at[r] = \frac{\log \vx[r] + \log \Gamma}{\lambda}$ for every coordinate $r \in [d]$, where $\Gamma = \sum_{\ind = 1}^{d} \exp\{\lambda \at[\ind]\}$ and $\vx \in \Delta^d$, we get
    \[
        \frac{\partial^2 f(\lambda; \at)}{\partial \lambda^2} & = \frac{1}{\lambda^2} \sum_{\ind = 1}^{d } \big( \vx[\ind] (\log \vx[\ind] + \log \Gamma)^2 \big) - \frac{1}{\lambda^2} \Big ( \sum_{\ind =1}^{k}  \vx[\ind] (\log \vx[\ind] + \log \Gamma) \Big)^2 -\frac{\alpha - 1}{\lambda^2} \\
        & = \frac{1}{\lambda^2} \sum_{\ind = 1}^{d } \big( \vx[\ind] (\log \vx[\ind] + \log \Gamma)^2 \big) - \frac{1}{\lambda^2} \Big ( \sum_{\ind =1}^{k}  \vx[\ind] \log \vx[\ind]  + \log \Gamma \Big)^2 -\frac{\alpha - 1}{\lambda^2}  \\
        & = \frac{1}{\lambda^2} \mleft( \sum_{\ind =1}^{d} \vx[\ind] \log^2 \vx[\ind] - \Big( \sum_{\ind =1}^{k} \vx[\ind] \log \vx[\ind] \Big)^2 - (\alpha - 1) \mright)                                                                   \\
        & \leq \frac{1}{\lambda^2} \mleft( \log^2 d - (\alpha - 1) \mright).
    \]
    Therefore, the function $f(\lambda; \at)$ is $\frac{\alpha - \log^2 d - 1}{\lambda^2}$-strongly concave.
\end{proof}
Secondly, we show that the optimization problem \eqref{eq:opt_problem_lambda} is self-concordant. 
\begin{lemma}[Self-concordance of learning rate control step] \label{lemma:self_concordance}
    The objective of the optimization problem,
    \[
        \hat{\lambda} \defeq \argmax_{\lambda \in (0, \eta]} \left\{f(\lambda; \at) \defeq (\alpha - 1) \log \lambda + \log \Big( {\sum_{\ind=1}^d e^{\lambda \at[\ind]}} \Big)  \right\},
    \]
    is self-concordant., i.e.,
    \[
        \mleft( \frac{\partial^3 f(\lambda; \at)}{\partial \lambda^3} \mright)^2 \leq  - 4 \mleft(\frac{\partial^2 f(\lambda; \at)}{\partial \lambda^2} \mright)^3.
    \]
\end{lemma}
\begin{proof}
    Recall that from proof of \Cref{lemma:convexity_step},
    \[
        \frac{\partial^2 f(\lambda; \at)}{\partial \lambda^2} & = \frac{\sum_{\ind = 1}^d \at[\ind]^2 e^{\lambda \at[\ind]} }{\sum_{\ind = 1}^d e^{\lambda \at[\ind]}} - \mleft( \frac{\sum_{\ind = 1}^d \at[\ind] e^{\lambda \at[\ind]}}{\sum_{\ind = 1}^d e^{\lambda \at[\ind]}}\mright)^2 - \frac{\alpha - 1}{\lambda^2},
    \]
    Furthermore,
    \[
        - \frac{\partial^2 f(\lambda; \at)}{\partial \lambda^2} \geq \frac{\alpha - \log^2 d - 1}{\lambda^2}.
    \]
    Taking derivatives,
    \[
        \frac{\partial^3 f(\lambda; \at)}{\partial \lambda^3} & = \frac{\sum_{\ind = 1}^d \at[\ind]^3 e^{\lambda \at[\ind]} }{\sum_{\ind = 1}^d e^{\lambda \at[\ind]}} - \frac{\Big( \sum_{\ind = 1}^d \at[\ind]^2 e^{\lambda \at[\ind]} \Big) \Big(  \sum_{\ind = 1}^d \at[\ind] e^{\lambda \at[\ind]}\Big)}{\Big(\sum_{\ind = 1}^d e^{\lambda \at[\ind]}\Big)^2}                                                                                                      \\
        & \quad - 2 \mleft( \frac{\sum_{\ind = 1}^d \at[\ind] e^{\lambda \at[\ind]}}{\sum_{\ind = 1}^d e^{\lambda \at[\ind]}} \mright) \mleft( \frac{\sum_{\ind = 1}^d \at[\ind]^2 e^{\lambda \at[\ind]} }{\sum_{\ind = 1}^d e^{\lambda \at[\ind]}} - \mleft( \frac{\sum_{\ind = 1}^d \at[\ind] e^{\lambda \at[\ind]}}{\sum_{\ind = 1}^d e^{\lambda \at[\ind]}}\mright)^2 \mright) + \frac{2 (\alpha - 1)}{\lambda^3}.
    \]
    Substituting $\at[r] = \frac{\log \vx[r] + \log \Gamma}{\lambda}$ for every coordinate $r \in [d]$, where $\Gamma = \sum_{\ind = 1}^{d} \exp\{\lambda \at[\ind]\}$ and $\vx \in \Delta^d$, we get
    \[
        & \frac{\partial^3 f(\lambda; \at)}{\partial \lambda^3}                                                                                                                                                                                                                                     \\
        & \quad  = \frac{1}{\lambda^3} \sum_{\ind = 1}^{d } \vx[\ind] (\log \vx[\ind] + \log \Gamma)^3 - 3 \frac{1}{\lambda^3} \Big( \sum_{\ind = 1}^{d } \vx[\ind] (\log \vx[\ind] + \log \Gamma)^2 \Big)  \Big( \sum_{\ind = 1}^{d } \vx[\ind] (\log \vx[\ind] + \log \Gamma) \Big)  \\
        & \quad \qquad + 2 \frac{1}{\lambda^3} \Big( \sum_{\ind = 1}^{d } \vx[\ind] (\log \vx[\ind] + \log \Gamma) \Big)^3 + \frac{2 (\alpha - 1)}{\lambda^3}                                                                                                                          \\
        & \quad = \frac{1}{\lambda^3} \mleft( \sum_{\ind = 1}^{d} \vx[k]\log^3 \vx[k] - 3 \big( \sum_{\ind = 1}^{d} \vx[k] \log^2 \vx[k] \big) \big( \sum_{\ind = 1}^{d} \vx[k] \log \vx[k] \big) + 2 \big ( \sum_{\ind = 1}^{d} \vx[k] \log \vx[k] \big)^3  + 2 (\alpha - 1) \mright) \\
        & \quad \leq \frac{1}{\lambda^3} \mleft( 3 \log^3 d +  2 (\alpha - 1) \mright)
    \]
    Thus,
    \[
        \mleft(\frac{\partial^3 f(\lambda; \at)}{\partial \lambda^3}\mright)^2 \leq \mleft(\frac{3 \log^3 d +  2 (\alpha - 1)}{\lambda^3}  \mright)^2 \leq 4 \mleft( \frac{\alpha - \log^2 d - 1}{\lambda^2} \mright)^3 \leq
        - 4 \mleft( \frac{\partial^2 f(\lambda; \at)}{\partial \lambda^3} \mright)^3,
    \]
    for all $d \geq 1$,\footnote{the equality happens at $d = 1$.} where $\alpha = \beta \log^2 d + 2 \log d + 2$ and $\beta \geq 2$, hence the proof is concluded.
\end{proof}

\begin{restate}{lemma:convexity_concordance_step}[Properties of learning rate control step]
    For any $\at \in \bbR^d$, the rate control objective $f(\lambda; \at)$, defined in
    \[
        \hat{\lambda} \defeq \argmax_{\lambda \in (0, \eta]} \left\{f(\lambda; \at) \defeq (\alpha - 1) \log \lambda + \log \Big( {\sum_{\ind=1}^d e^{\lambda \at[\ind]}} \Big)  \right\},
    \]
    satisfies the following properties:
    \begin{itemize}
        \item Strong concavity: $f''(\lambda; \at) \leq -(\alpha - \log^2 d - 1)/\lambda^2$ for all $\lambda \in (0,\infty)$.
        \item Self-concordance: $(f'''(\lambda; \at))^2 \leq -4 f''(\lambda; \at)^{3}$,
    \end{itemize}
    where all derivatives are with respect to $\lambda$.
\end{restate}
\begin{proof}
    We show strong concavity in \Cref{lemma:convexity_step} and self-concordance in \Cref{lemma:self_concordance}.
\end{proof}

Next, we need to prove \Cref{theorem:stability_lambda}. To do so, we first prove and state the following lemma, which demonstrate that when the the maximum regret $\max_k \{\at\^t[k]\}$ accumulated on the actions is not too negative, the optimal solution to optimization problem \eqref{eq:opt_problem_lambda} is $\lambda\^t = \eta$.

\begin{lemma}\label{lemma:lambda_is_one}
    Given an arbitrary vector $\at\in \bbR^d$, and $\hat{\lambda}$ as the solution to
    \begin{align}
        \hat{\lambda} \defeq \argmax_{\lambda \in (0, \eta]} \left\{f(\lambda; \at) \defeq (\alpha - 1) \log \lambda + \log \Big( {\sum_{\ind=1}^d e^{\lambda \at[\ind]}} \Big)  \right\},
    \end{align}
    then as long as $\max_{r \in [d]} \{ \at[r]\} \geq - \beta \log^2 d$, we have $\hat{\lambda} = \eta$.
\end{lemma}
\begin{proof}
    By KKT conditions, concavity and uniqueness of the solution, it is obvious that whenever $f^\prime (\eta; \at) \geq 0$, then $\widehat{\lambda} = \eta$. Thus, with the same arguments as proof of \Cref{lemma:stability},
    \[
        f^\prime(\eta; \at) & = \frac{\alpha - 1}{\eta} + \frac{1}{\eta} \ent{\vx} + \frac{1}{\eta} \log \mleft( \sum_{\ind = 1}^{d} e^{\lambda \at[\ind]}\mright) \\
        & \geq \frac{\alpha - 1}{\eta} - \frac{\log d}{\eta} + \frac{1}{\eta} \max_{r \in [d]} \{ \at[r]\}                                     \\
        & \geq \frac{1}{\eta} \mleft( \beta \log^2 d + \max_{r \in [d]} \{ \at[r]\} \mright).
    \]
    Thus, whenever $\max_{r \in [d]} \{ \at[r]\} \geq - \beta \log^2 d$, then $\hat{\lambda} = \eta$.
\end{proof}

In sequel, to prove \Cref{theorem:stability_lambda} (sensitivity of learning rates on regrets), we state and prove the following lemma on stability of learning rates in multiplicative sense. Combining a good analytic guess $\lambda_0$ for the value of $\lambda$ with techniques for self-concordant function analysis, we establish that the intrinsic norm of the second-order ascent direction of $f$ is small at $\lambda_0$. This allows us to conclude that the solution $\lambda$ must be within a small radius from $\lambda_0$ in intrinsic norm. Furthermore, using the bound on $f''(\lambda_0;\cdot)$ given in \cref{lemma:convexity_concordance_step}, we can finally conclude proximity in the multiplicative sense.

\begin{lemma}(Multiplicative Stability)\label{lemma:stability}
    Given vectors $\at, \aprimet \in \bbR^d$, then $\widehat{\lambda}, \widehat{\lambda^\prime}$, the solutions to
    \begin{align}
        \widehat{\lambda} = \argmax_{\lambda \in (0, \eta]} \left\{f(\lambda; \at) \defeq (\alpha - 1) \log \lambda + \log \Big( {\sum_{\ind=1}^d e^{\lambda \at[\ind]}} \Big)  \right\},
    \end{align}
    and
    \begin{align}
        \widehat{\lambda^\prime} = \argmax_{\lambda \in (0, \eta]} \left\{f(\lambda; \aprimet) \defeq (\alpha - 1) \log \lambda + \log \Big( {\sum_{\ind=1}^d e^{\lambda \aprimet[\ind]}} \Big)  \right\}
    \end{align}
    respectively, are multiplicatively stable, i.e.,
    \[
        \frac{\beta - 5}{\beta + 3} \mleft(\frac{\min_{r \in d}{(- \aprimet[r]})}{\min_{r \in d}{(- \at[r]})} \mright) \leq \frac{\widehat{\lambda}}{\widehat{\lambda^\prime}} \leq \frac{\beta + 3}{\beta - 5} \mleft( \frac{\min_{r \in d}{(- \aprimet[r]})}{\min_{r \in d}{(- \at[r]})} \mright).
    \]
\end{lemma}
\begin{proof}
    We use suboptimality bound entailed by Newton update for self-concordance functions to prove stability. Recall that from proof of \Cref{lemma:convexity_step},
    \[
        \frac{\partial f(\lambda; \at)}{\partial \lambda}      & = \frac{\sum_{\ind = 1}^d \at[\ind] e^{\lambda \at[\ind]}}{\sum_{\ind = 1}^d e^{\lambda \at[\ind]}} + \frac{\alpha - 1}{\lambda}.                                                                                                                            \\
        \frac{\partial^2 f(\lambda; \at)}{\partial \lambda^2} & = \frac{\sum_{\ind = 1}^d \at[\ind]^2 e^{\lambda \at[\ind]} }{\sum_{\ind = 1}^d e^{\lambda \at[\ind]}} - \mleft( \frac{\sum_{\ind = 1}^d \at[\ind] e^{\lambda \at[\ind]}}{\sum_{\ind = 1}^d e^{\lambda \at[\ind]}}\mright)^2 - \frac{\alpha - 1}{\lambda^2}.
    \]
    Let us choose the primary guess $\lambda_0 = \dfrac{\alpha - 1}{\min_{r \in [d]}{(- \at[r]})}$. By change of variables, we know that
    \[
        \frac{\partial f}{\partial \lambda}(\lambda_0; \at) = \frac{\alpha - 1}{\lambda_0} + \frac{1}{\lambda_0} \ent{\vx} + \frac{1}{\lambda_0} \log \mleft( \sum_{\ind = 1}^{d} e^{\lambda_0 \at[\ind]}\mright),
    \]
    where $\ent{\vx}$ is the negative entropy of the vector $\vx [r] = \dfrac{\exp\{\lambda_0 \at[r]\}}{\sum_{\ind = 1}^{d} \exp\{\lambda_0 \at[\ind]\}}$. Therefore, $0 \geq H(\vx) \geq - \log d$.
    On the other hand, by \Cref{lemma:softmax_bound},
    \[
        \max_{r \in [d]} \{\at[r]\} \leq \frac{1}{\lambda_0} \log \mleft( \sum_{\ind = 1}^{d} e^{\lambda_0 \at[\ind]}\mright) \leq  \max_{r \in [d]} \{\at[r]\} + \frac{\log d}{\lambda_0}.
    \]
    Hence,
    \[
        - \frac{\log d}{\lambda_0} = \min_{r \in [d]} \{-\at[r]\} +  \max_{r \in [d]} \{\at[r]\}- \frac{\log d}{\lambda_0} \leq f^\prime(\lambda_0, \at) \leq  \min_{r \in [d]} \{-\at[r]\} +  \max_{r \in [d]} \{\at[r]\} + \frac{\log d}{\lambda_0} = \frac{\log d}{\lambda_0}.
    \]
    On the other hand, from \Cref{lemma:convexity_step} we know that
    \[
        | f^{\dprime}(\lambda_0; \at) | \geq \frac{1}{\lambda_0^2} \mleft( \alpha - \log^2 d  - 1 \mright).
    \]
    Therefore, the local norm of Newton step $n(\lambda_0)$ is
    \[
        \| n(\lambda_0) \|_{f^{\dprime}(\lambda_0; \at)}^2 & = \frac{f^\prime(\lambda_0, \at)^2}{|f^{\dprime}(\lambda_0, \at) |} \\
        & \leq \frac{\log^2 d}{(\beta - 1) \log^2 d + 2 \log d + 1}           \\
        & \leq \frac{1}{\beta - 1},
    \]
    controlled by the hyperparameter $\beta \geq 17$.\footnote{\textrm{recall that $\alpha = \beta \log^2 d + 2 \log d + 2$.}} In turn, by \Cref{theorem:newton_self},
    \[
        & \| \hat{\lambda} - \lambda_0 \|_{f^{\dprime}(\lambda_0; \at)} \leq \frac{3}{\sqrt{\beta - 1}}                                       \\
        \Rightarrow \quad & | \hat{\lambda} - \lambda_0 | \leq \frac{3}{\sqrt{(\beta - 1) |f^{\dprime} (\lambda_0; \at) |}}  \leq \frac{4 \lambda_0}{\beta - 1} \\
        \Rightarrow \quad & \mleft| \frac{\hat{\lambda}}{\lambda_0} - 1 \mright| \leq \frac{4}{\beta - 1} \numberthis{eq:stab_lambda},
    \]
    for all $d \geq 2$.
    In a similar manner, we choose the primary guess $\lambda^\prime_0 = \dfrac{\alpha - 1}{\min_{r \in d}{\{- \aprimet[r]}\}}$ and we infer,
    \[
        \mleft| \frac{\hat{\lambda^\prime}}{\lambda^\prime_0} - 1 \mright| \leq \frac{4}{\beta - 1}. \numberthis{eq:stab_lambda_prime}
    \]
    Combining \Cref{eq:stab_lambda,eq:stab_lambda_prime},
    \[
        & (1 - \frac{4}{\beta - 1}) \lambda_0 \leq \hat{\lambda} \leq (1 + \frac{4}{\beta - 1}) \lambda_0 \numberthis{eq:one_lambda}                                                   \\
        & (1 - \frac{4}{\beta - 1}) \lambda^\prime_0 \leq \hat{\lambda^\prime} \leq (1 + \frac{4}{\beta - 1}) \lambda^\prime_0                                                     \\
        \Rightarrow & \frac{\beta - 5}{\beta + 3} \mleft(\frac{\min_{r \in d}{(- \aprimet[r]})}{\min_{r \in d}{(- \at[r]})} \mright) = \frac{\beta - 5}{\beta + 3} . \frac{\lambda_0}{\lambda^\prime_0}  \leq \frac{\hat{\lambda}}{\hat{\lambda^\prime}} \leq \frac{\beta + 3}{\beta - 5}.\frac{\lambda_0}{\lambda^\prime_0} = \frac{\beta + 3}{\beta - 5} \mleft( \frac{\min_{r \in d}{(- \aprimet[r]})}{\min_{r \in d}{(- \at[r]})} \mright).
    \]
\end{proof}

Now, we are ready to state and prove \Cref{theorem:stability_lambda}. To this end, we consider three cases based on whether $\max_{r \in [d]} \{ \at[r]\}$ and $\max_{r \in [d]} \{ \aprimet[r]\}$ are higher or lower than $- \beta \log^2 d$. In the regimes, where \Cref{lemma:lambda_is_one} is satisfied the proof is immediate. For other cases, we use careful analysis with the help of \Cref{lemma:stability}.

\begin{restate}{theorem:stability_lambda} (Sensitivity of learning rates on regrets)
    There exists a universal constant $\beta$,\footnote{For concrete values, choosing any $\beta \geq 70$ suffices.} such that for $\alpha \ge 2 + 2 \log d + \beta \log^2 d$, the following property holds.
    Let $\at, \aprimet \in \bbR^d$ be such that $\| \at - \aprimet \|_{\infty} \leq 2$, and let $\hat{\lambda}, \hat{\lambda'}$ the corresponding learning rates, that is,
    \[
        \hat{\lambda} = \argmax_{t \in (0, \eta]} f(t; \at), \qquad
        \hat{\lambda'} = \argmax_{t \in (0,\eta]} f(t; \at').
    \]
    Then, $\hat{\lambda}$ and $\hat{\lambda'}$ are multiplicatively stable; specifically,
    \[
        \frac{7}{10} \leq  \frac{\hat{\lambda}}{\hat{\lambda^\prime}} \leq \frac{7}{5},
    \]
\end{restate}
\begin{proof}
    Here we show the lemma for $\| \at - \aprimet \|_{\infty} \leq 2$. The extension to general $\| \at - \aprimet \|_{\infty} \leq o(1)$ is easy to infer by choosing $\beta$ large enough. We have three cases depending on the size of $\max_{r \in [d]} \{ \at[r]\}$.

    \begin{enumerate}
        \item If $\max_{r \in [d]} \{ \at[r]\} \geq - \beta \log^2 d$ and $\max_{r \in [d]} \{ \aprimet[r]\} \geq - \beta \log^2 d$, then by \Cref{lemma:lambda_is_one}, we conclude that $\hat{\lambda} = \hat{\lambda^\prime} = \eta$.

        \item If $\max_{r \in [d]} \{ \at[r]\} < - \beta \log^2 d$ and $\max_{r \in [d]} \{ \aprimet[r]\} < - \beta \log^2 d$, then by \Cref{lemma:stability} and $\| \at - \aprimet \|_{\infty} \leq 2$,
              \[
                  \frac{4}{5} \frac{\beta - 5}{\beta + 3} \leq \frac{\beta - 5}{\beta + 3} \mleft( \frac{\beta \log^2 d}{2 + \beta \log^2 d} \mright) \leq \frac{\hat{\lambda}}{\hat{\lambda^\prime}} \leq \frac{\beta + 3}{\beta - 5} \mleft( \frac{2 + \beta \log^2 d}{\beta \log^2 d}\mright) \leq \frac{6}{5} \frac{\beta + 3}{\beta - 5},
              \]
              for $\beta \geq 20$ and $d \geq 2$.

        \item And otherwise, if without loss of generality assume that $\max_{r \in [d]} \{ \at[r]\} \geq - \beta \log^2 d$ and  $\max_{r \in [d]} \{ \aprimet[r]\} < - \beta \log^2 d$. Then, $\hat{\lambda} = 1$. On the other hand, based on \Cref{eq:one_lambda} in the proof of \Cref{lemma:stability},
              \[
                  \mleft(1 - \frac{4}{\beta - 1}\mright) \lambda^\prime_0 \leq \hat{\lambda^\prime} \leq \mleft(1 + \frac{4}{\beta - 1} \mright) \lambda^\prime_0.
              \]
              Hence,
              \[
                  \frac{\alpha - 1}{2 + \beta \log^2 d} \frac{\beta - 5}{\beta - 1}
                  \leq \frac{\beta - 5}{\beta - 1} \dfrac{\alpha - 1}{\min_{r \in d}\{- \alpha^{\prime}[r]\}}
                  \leq \frac{\hat{\lambda^\prime}}{\hat{\lambda}}
                  \leq \frac{\beta + 3}{\beta - 1} \dfrac{\alpha - 1}{\min_{r \in d}{\{- \alpha^{\prime}[r]\}}}
                  \leq \dfrac{\alpha - 1}{\beta \log^2 d} \frac{\beta + 3}{\beta - 1}
              \]
              since $\| \at - \aprimet \|_{\infty} \leq 2$ and thus $\min_{r \in d}{\{- \aprimet[r]}\} \leq 2 + \beta \log^2 d$. Consequently,
              \[
                  \frac{4}{5} \frac{\beta - 5}{\beta - 1} \leq \frac{\alpha - 1}{2 + \beta \log^2 d} \frac{\beta - 5}{\beta - 1}
                  \leq \frac{\hat{\lambda^\prime}}{\hat{\lambda}}
                  \leq \dfrac{\alpha - 1}{\beta \log^2 d} \frac{\beta + 3}{\beta - 1}
                  \leq \frac{6}{5} \frac{\beta + 3}{\beta - 1}
              \]
              for $\beta \geq 50$ and $d\geq 2$,
    \end{enumerate}
    Putting all cases together, for $\beta \geq 50$,
    \[
        \frac{4}{5} \frac{\beta - 5}{\beta + 3} \leq \frac{\hat{\lambda^\prime}}{\hat{\lambda}}
        \leq \frac{6}{5} \frac{\beta + 3}{\beta - 5}.
    \]
    Now, by choosing any $\beta \geq 70$,
    \[
        \frac{7}{10} = \frac{4}{5} \times \frac{7}{8} \leq \frac{\hat{\lambda^\prime}}{\hat{\lambda}}
        \leq \frac{6}{5} \times \frac{8}{7} < \frac{7}{5} \quad \Rightarrow \quad \frac{\hat{\lambda^\prime}}{\hat{\lambda}}  \in \mleft[\frac{21}{20} - \frac{7}{20}, \frac{21}{20} + \frac{7}{20}\mright],
    \]
    and the proof is completed.
\end{proof}

\subsubsection{Equivalent Viewpoints of \ours}
\label{sec:equivalent_viewpoints}

In this section, we discuss equivalent formulations of \ours. First, to conceptualize the idea of dynamic learning rate control, we show that \ours (as presented in \Cref{algo:ours_lambdaview}) is a special case of Formulation \ref{eq:DRLC_FTRL} with the choice of regularizers $\rho(\lambda) = (\alpha - 1) \log \lambda$ and $\phi(\vx) = \sum_{\ind=1}^d \vx[\ind] \log \vx[\ind]$, making it a special, easily computable instance. Secondly, we show that, for the purpose of regret analysis, \ours is equivalent to \OFTRL with the regularizer $\psi$ on the space $(0, 1] \Delta^n$ defined in \Cref{eq:regularizer}. This equivalence allows us to leverage standard techniques to analyze the regret, given the strong spectral properties of the regularizer $\psi$, which we will establish later.

\begin{restate}{theorem:equivalence_viewpoints}
    \Cref{algo:ours_lambdaview} (\ours) can alternatively be viewed as \oursgeneral (Formulation \ref{eq:DRLC_FTRL}) with the choice of regularizers $\rho(\lambda) = (\alpha - 1) \log \lambda$ and $\phi(\vx) = \sum_{\ind=1}^d \vx[\ind] \log \vx[\ind]$, i.e.,
    \[
        \vstack{\lambda\^t}{\vx\^t} \gets \argmax_{\lambda \in (0, \eta], \vx \in \Delta^d}\mleft\{ {\lambda} \mleft\langle  \at\^t, \vx \mright\rangle + (\alpha - 1) \log \lambda - \sum_{\ind=1}^d \vx[\ind] \log \vx[\ind]\mright\}.
    \]
    or additionally as \OFTRL with the regularizer $\psi$ defined in \Cref{eq:regularizer}, i.e.,
    \[
        \vy\^t \gets \argmax_{\vy \in (0,1]\Delta^d }\mleft\{ {\eta} \mleft\langle  \at\^t, \vy \mright\rangle + \alpha \log \left(\sum_{\ind=1}^d \vy[\ind]\right) - \frac{1}{\sum_{\ind=1}^d \vy[\ind]} \sum_{\ind=1}^d \vy[\ind] \log \vy[\ind]\mright\},
    \]
    In other words, the three perspectives (\Cref{algo:ours_lambdaview}, Formulations \ref{eq:drlc_ftrl_special} and \ref{eq:oftrl_psi}) of \ours are equivalent and result in the same learning dynamics.
\end{restate}

\begin{proof}
    The proof follows by combining \Cref{theorem:equivalence_opt,lemma:KKT,corollary:lambda_softmax}.
\end{proof}

\begin{proposition}
\label{theorem:equivalence_opt}
    The optimization step of the following \OFTRL,
    \[
        \vy\^t \gets \argmax_{\vy \in (0,1]\Delta^d }\mleft\{ {\eta} \mleft\langle  {\Ut}\^t +  {\vec\ut}\^{t-1}, \vy \mright\rangle + \alpha \log \Big(\sum_{\ind=1}^d \vy[\ind]\Big) - \frac{1}{\sum_{\ind=1}^d \vy[\ind]} \sum_{\ind=1}^d \vy[\ind] \log \vy[\ind]\mright\} \numberthis{algo:ours_yview}
    \]
    is equivalent to optimization step of the following special instance of Formulation \ref{eq:DRLC_FTRL},
    \[
        \vstack{\lambda\^t}{\vx\^t} \gets \argmax_{\lambda  \in (0, 1], \vx \in \Delta^d}\mleft\{ {\eta} {\lambda} \mleft\langle  {\Ut}\^t +  {\vec\ut}\^{t-1}, \vx \mright\rangle + (\alpha - 1) \log \lambda - \sum_{\ind=1}^d \vx[\ind] \log \vx[\ind]\mright\} \numberthis{algo:ours_xview}
    \]
    under the change of variables $\lambda\^t = \sum_{\ind = 1}^d \vy[\ind]$ and  $\vx\^t[r] = \frac{\vy[r]}{\sum_{\ind = 1}^d \vy[\ind]}$.
\end{proposition}
\begin{proof}
    By choice of $\vy = \lambda \vx$, for any feasible point in the first optimization problem,
    \[
        & \eta \left\langle  \Ut\^{t} + \ut\^{t-1}, \lambda \vx \right\rangle + \alpha \log \Big(\sum_{\ind=1}^d \lambda \vx[\ind]\Big) - \frac{1}{\sum_{\ind=1}^d \lambda \vx[\ind]} \sum_{\ind=1}^d \lambda \vx[\ind] \log \left(\lambda \vx[\ind]\right) \\
        & \quad = \lambda \left\langle  \Ut\^{t} + \ut\^{t-1}, \vx \right\rangle + \alpha \log (\lambda) + \alpha \log \Big(\sum_{\ind=1}^d \vx[\ind]\Big) - \frac{1}{\sum_{\ind=1}^d \vx[\ind]} \sum_{\ind=1}^d \vx[\ind] (\log \vx[\ind] + \log \lambda)  \\
        & \quad = {\eta} {\lambda} \mleft\langle  {\Ut}\^t +  {\vec\ut}\^{t-1}, \vx \mright\rangle + (\alpha - 1) \log \lambda - \sum_{\ind=1}^d \vx[\ind] \log \vx[\ind],
    \]
    where the last line follows since $\sum_{\ind=1}^d \vx[\ind] = 1$. The other direction holds similarly.
\end{proof}
\begin{lemma} \label{lemma:KKT}
    The optimal solution $(\lambda\^t, \vx\^t)$ to the following special instance of \oursgeneral,
    \[
        \vstack{\lambda\^t}{\vx\^t} \gets \argmax_{\lambda  \in (0, \eta], \vx \in \Delta^d}\mleft\{{\lambda} \mleft\langle  {\Ut}\^t +  {\vec\ut}\^{t-1}, \vx \mright\rangle + (\alpha - 1) \log \lambda - \sum_{\ind=1}^d \vx[\ind] \log \vx[\ind]\mright\} \numberthis{eq:opt_2}
    \]
    satisfies,\footnote{Note that this formulation is the same as \Cref{algo:ours_xview}. It is entailed by simply scaling $\lambda$ by $\eta$.}
    \[
        \vx\^t[r] = \textup{softmax}(\lambda\^t \at\^t)[r] = \frac{\exp\{\lambda\^t \at\^t[r]\}}{\sum_{\ind = 1}^{d} \exp\{\lambda\^t \at\^t[\ind]\}} \numberthis{eq:x_as_lambda},
    \]
    for every coordinate $r \in [d]$, where $\at\^t[r] = {\Ut}\^t[r] + {\vec\ut}\^{t-1}[r]$.
\end{lemma}
\begin{proof}
    By KKT conditions,
    \[
        \lambda\^t \at\^t - \big[\log x\^t[1], \log x\^t[2], \cdots, \log x\^t[d]\big]^\top = \mu \in \bbR.
    \]
    Therefore, $\vx\^t[r] \propto \exp\{\at\^t[r]\}$. Now, after renormalization since $\vx \in \Delta^d$,
    \[
        \vx\^t[r] = \frac{\exp\{\lambda\^t \at\^t[r]\}}{\sum_{\ind = 1}^{d} \exp\{\lambda\^t \at\^t[\ind]\}}.
    \]
\end{proof}
\begin{corollary} \label{corollary:lambda_softmax}
    We know that the optimal solution $\lambda\^t$ can be written as,
    \[
        \lambda\^t = \argmax_{\lambda \in (0, \eta]} \left\{ \log \Big( {\sum_{\ind=1}^d e^{\lambda \at\^t[\ind]}} \Big) + (\alpha - 1) \log \lambda \right\}, \numberthis{eq:opt_problem_lambda_temp2}
    \]
    where $\at\^t[r] = {\Ut}\^t[r] + {\vec\ut}\^{t-1}[r]$.
\end{corollary}
\begin{proof}
    Given \Cref{lemma:KKT}, plugging \Cref{eq:x_as_lambda} into \Cref{eq:opt_2} entails,
    \[
        & \argmax_{\lambda \in [0, \eta]}                                                                                          \\
        & \quad  \left\{ \lambda \sum_{r=1}^d \at\^t[r] \frac{e^{\lambda \at\^t[r]}}{\sum_{\ind=1}^d e^{\lambda \at\^t[\ind]}} + (\alpha - 1) \log \lambda - \sum_{r=1}^d \frac{e^{\lambda \at\^t[r]}}{\sum_{r=1}^d e^{\lambda \at\^t[\ind]}} \log \Big( \frac{e^{\lambda \at\^t[r]}}{\sum_{\ind=1}^d e^{\lambda \at\^t[\ind]}} \Big) \right\}                                                                               \\
        & = \argmax_{\lambda \in [0, \eta]}                                          \\
        & \quad  \left\{ \lambda \sum_{r=1}^d \at\^t[r] \frac{e^{\lambda \at\^t[r]}}{\sum_{\ind=1}^d e^{\lambda \at\^t[\ind]}} + (\alpha - 1) \log \lambda - \lambda \sum_{r=1}^d \at\^t[r] \frac{e^{\lambda \at\^t[r]}}{\sum_{\ind=1}^d e^{\lambda \at\^t[\ind]}} + \frac{ \sum_{r=1}^d  e^{\lambda \at\^t[r]}}{\sum_{\ind=1}^d e^{\lambda \at\^t[\ind]}} \log \Big( {\sum_{\ind=1}^d e^{\lambda \at\^t[\ind]}} \Big) \right\} \\
        & = \argmax_{\lambda \in [0, \eta]} \left\{ \log \Big( {\sum_{\ind=1}^d e^{\lambda \at\^t[\ind]}} \Big) + (\alpha - 1) \log \lambda \right\}.                                                                    \\
    \]
\end{proof}

\subsubsection{Strong Spectral Properties of $\psi$} \label{sec:spectral_psi}
In this section, we prove strong convexity of $\psi$ and hence we demonstrate that Bregman divergence of $\psi$ is well-defined.

\begin{restate}{theorem:regularizer_convex}
    The regularizer $\psi (\vy)$ is strongly convex and furthermore,
    \[
        \nabla^2 \psi(\vy) \succeq \frac{1}{2 \eta} \begin{pmatrix}
            \frac{1}{\vy[1] \cdot \sum_{\ind=1}^d \vy[\ind]} & 0                                                & \cdots & 0                                                \\
            0                                                & \frac{1}{\vy[2] \cdot \sum_{\ind=1}^d \vy[\ind]} & \cdots & 0                                                \\
            \vdots                                           & \vdots                                           & \ddots & \vdots                                           \\
            0                                                & 0                                                & \cdots & \frac{1}{\vy[d] \cdot \sum_{\ind=1}^d \vy[\ind]}
        \end{pmatrix}.
    \]
\end{restate}
\begin{proof}
    The partial derivatives of the regularizer are
    \[
        \eta \frac{\partial \psi}{\partial \vy[i]} = - \frac{\alpha}{\sum_{\ind=1}^d \vy[\ind]} - \frac{1}{(\sum_{\ind=1}^d \vy[\ind])^2} \left(\sum_{\ind=1}^d \vy[\ind] \log \vy[\ind]\right) + \frac{1 + \log \vy[i]}{\sum_{\ind=1}^d \vy[\ind]}
    \]
    For convenience, let $\summ{\vy} \defeq \sum_{\ind=1}^d \vy[\ind]$ and $\vx[i] \defeq \frac{\vy[i]}{\summ{\vy}}$, so that $\vx \in \Delta^d$. Let $\mathds{1}_{i = j}$ be the indicator function that is evaluated to one when $i=j$ and zero elsewhere. The second order derivatives are as follows,
    \[
        \eta \frac{\partial^2 \psi}{\partial \vy[i] \partial \vy[j]} & = \frac{\alpha}{\summ{\vy}^2} + \frac{2}{\summ{\vy}^3} \left(\sum_{\ind=1}^d \vy[\ind] \log \vy[\ind]\right) - \frac{1 + \log \vy[j]}{\summ{\vy}^2} - \frac{1 + \log \vy[i]}{\summ{\vy}^2} + \frac{\mathds{1}_{i = j}}{\vy[i] \summ{\vy}}                                 \\
        & = \frac{\alpha}{\summ{\vy}^2} + \frac{2}{\summ{\vy}^2} \left(\sum_{\ind=1}^d \vx[\ind] \log (\summ{\vy}\vx[\ind])\right) - \frac{2}{\summ{\vy}^2} \\&\hspace{6cm}- \frac{\log (\summ{\vy}\vx[i]) + \log (\summ{\vy}\vx[j])}{\summ{\vy}^2}+ \frac{\mathds{1}_{i = j}}{\vy[i] \summ{\vy}} \\
        & = \frac{\alpha - 2 + 2 \sum_{\ind=1}^d \vx[\ind] \log \vx[\ind]}{\summ{\vy}^2} + \frac{2 \log \summ{\vy}}{\summ{\vy}^2} - \frac{2 \log \summ{\vy}}{\summ{\vy}^2} \\&\hspace{6cm}- \frac{\log \vx[i] + \log \vx[j]}{\summ{\vy}^2} + \frac{\mathds{1}_{i = j}}{\vy[i] \summ{\vy}}          \\
        & = \frac{\alpha - 2 + 2 \sum_{\ind=1}^d \vx[\ind] \log \vx[\ind]}{\summ{\vy}^2} - \frac{\log \vx[i] + \log \vx[j]}{\summ{\vy}^2} + \frac{\mathds{1}_{i = j}}{\vy[i] \summ{\vy}}.
    \]\allowdisplaybreaks
    Let now $\alpha \defeq 2 + 2 \log d + \alpha^\prime$; we can therefore guarantee that for any vector $\vec v \in \bbR^d$,
    \[
        & \eta \summ{\vy}^2 [\vec v^\top \nabla^2 \psi(\vy) \vec v]                                                                                                                                                                                                                                                                                       \\
        & \hspace{.5cm} \geq \alpha^\prime \mleft(\sum_{\ind = 1}^k \vec v[k] \mright)^2 + \mleft( \sum_{\ind = 1}^{d} \frac{\vec v[\ind]^2}{\vx[i]} \mright) + \mleft( \sum_{\ind = 1}^{d} - 2 \vec v[\ind] \log \vx[i] \mright) \mleft( \sum_{\ind = 1}^{d} \vec v[\ind] \mright)                                                                       \\
        & \hspace{.5cm} = \mleft( \sum_{\ind = 1}^{d} \frac{\vec v[\ind]^2}{2 \vx[i]} \mright) + \alpha^\prime \mleft(\sum_{\ind = 1}^k \vec v[k] \mright)^2 + \mleft( \sum_{\ind = 1}^{d} \frac{\vec v[\ind]^2}{2 \vx[i]} \mright) + \mleft( \sum_{\ind = 1}^{d} - 2 \vec v[\ind] \log \vx[i] \mright) \mleft( \sum_{\ind = 1}^{d} \vec v[\ind] \mright) \\
        & \hspace{.5cm} = \mleft( \sum_{\ind = 1}^{d} \frac{\vec v[\ind]^2}{2 \vx[i]} \mright) + \alpha^\prime \mleft(\sum_{\ind = 1}^k \vec v[k] \mright)^2                                                                                                                                                                                              \\
        & \hspace{1.5cm} + \sum_{\ind = 1}^{d} \mleft[ \mleft( \sqrt{\frac{1}{2 \vx[\ind]}} \vec v[\ind] - \left(\sum_{j = 1}^{d} \vec v[j]\right) \sqrt{2 \vx[\ind]} \log \vx[\ind]\mright)^2 - 2 \left(\sum_{j = 1}^{d} \vec v[\ind]\right)^2 \vx[\ind] \log^2 \vx[\ind]\mright]                                                                        \\
        & \hspace{.5cm} \geq \mleft( \sum_{\ind = 1}^{d} \frac{\vec v[\ind]^2}{2 \vx[i]} \mright) + \alpha^\prime \mleft(\sum_{\ind = 1}^k \vec v[k] \mright)^2 - 2 \mleft( \sum_{\ind = 1}^{d} \vec v[\ind] \mright)^2 \mleft(\sum_{\ind = 1}^{d} \vx[\ind] \log^2 x[\ind] \mright)                                                                      \\
        & \hspace{.5cm} \geq \mleft( \sum_{\ind = 1}^{d} \frac{\vec v[\ind]^2}{2 \vx[i]} \mright) + \alpha^\prime \mleft(\sum_{\ind = 1}^k \vec v[k] \mright)^2 - 2 \log^2 d \mleft(\sum_{\ind = 1}^k \vec v[k] \mright)^2.
    \]
    This shows that by setting any $\alpha^\prime \geq 2 \log^2 d$, we can obtain
    \[
        \eta \vec v^\top \nabla^2 \psi(\vy) \vec v \geq \frac{1}{2} \vec v^\top \begin{pmatrix}
            \frac{1}{\vy[1] \cdot \sum_{\ind=1}^d \vy[\ind]} & 0                                                & \cdots & 0                                                \\
            0                                                & \frac{1}{\vy[2] \cdot \sum_{\ind=1}^d \vy[\ind]} & \cdots & 0                                                \\
            \vdots                                           & \vdots                                           & \ddots & \vdots                                           \\
            0                                                & 0                                                & \cdots & \frac{1}{\vy[d] \cdot \sum_{\ind=1}^d \vy[\ind]}
        \end{pmatrix} \vec v,
    \]
    which concludes the proof.
\end{proof}

\begin{proposition} \label{prop:bregman_rep}
    The Bregman divergence $D_\psi(\cdot \,\big\|\, \cdot)$ induced by the regularizer $\psi(\cdot)$ has the following representation:
    \[
        \eta D_{\psi}(\vz \,\big\|\, \vy) = (\alpha - 1) D_{\log} (\summ{\vz} \,\big\|\, \summ{\vy}) + \frac{\summ{\vz}}{\summ{\vy}} \kl{\vtheta}{\vx} + \left(1 - \frac{\summ{\vz}}{\summ{\vy}}\right) \ent{\vtheta} - \left(1 - \frac{\summ{\vz}}{\summ{\vy}}\right) \ent{\vx},
    \]
    where $D_{\log}$ is the Bregman divergence induced by the log regularizer.
\end{proposition}
\begin{proof}
    We know that
    \[
        \eta \frac{\partial \psi}{\partial \vy[i]} & = - \frac{\alpha}{\sum_{\ind=1}^d \vy[\ind]} - \frac{1}{(\sum_{\ind=1}^d \vy[\ind])^2} (\sum_{\ind=1}^d \vy[\ind] \log \vy[\ind]) + \frac{1 + \log \vy[i]}{\sum_{\ind=1}^d \vy[\ind]}                                                        \\
        & = - \frac{\alpha}{\summ{\vy}} - \frac{1}{\summ{\vy}} \big(\sum_{\ind = 1}^d \vx[\ind] \log \vx[\ind] \big) - \frac{\log \summ{\vy}}{\summ{\vy}} + \frac{1}{\summ{\vy}} + \frac{\log \vx[i]}{\summ{\vy}} + \frac{\log \summ{\vy}}{\summ{\vy}} \\
        & = - \frac{\alpha - 1}{\summ{\vy}} - \frac{1}{\summ{\vy}} \big(\sum_{\ind = 1}^d \vx[\ind] \log \vx[\ind] \big) + \frac{\log \vx[i]}{\summ{\vy}}.
    \]
    For the Bregman divergence, by definition, we get that
    \[
        D_{\psi}(\vz \,\big\|\, \vy) & = \psi(\vz) - \psi(\vy) - \sum_{\ind = 1}^d \frac{\partial \psi (\vy) }{\partial \vy[\ind]} (\vz[\ind] - \vy[\ind]).
    \]
    Hence,
    \[
        \hspace{-0.6cm} \eta D_{\psi}(\vz \,\big\|\, \vy) & = \mleft( - (\alpha - 1) \log \summ{\vz} + \sum_{\ind = 1}^{d} \vtheta[\ind] \log \vtheta[\ind] \mright) - \mleft( - (\alpha - 1) \log \summ{\vy} + \sum_{\ind = 1}^{d} \vx[\ind] \log \vx[\ind] \mright)                                                                          \\
        & \qquad - \mleft( - \frac{\alpha - 1}{\summ{\vy}} - \frac{1}{\summ{\vy}} (\sum_{\ind = 1}^d \vx[\ind] \log \vx[\ind])\mright) (\summ{\vz} - \summ{\vy}) - \mleft( \frac{1}{\summ{\vy}} \sum_{\ind = 1}^d \log \vx[\ind] (\summ{\vz} \vtheta[\ind]- \summ{\vy} \vx[\ind]) \mright)                   \\
        & = - (\alpha - 1) \log \frac{\summ{\vz}}{\summ{\vy}} + (\alpha - 1) \left(\frac{\summ{\vz}}{\summ{\vy}} - 1\right) + \left(\frac{\summ{\vz}}{\summ{\vy}} - 1\right)\left(\sum_{\ind = 1}^{d} \vx[\ind] \log \vx[\ind]\right) + \sum_{\ind = 1}^{d} \vtheta[\ind] \log \vtheta[\ind] \\
        & \qquad - \frac{\summ{\vz}}{\summ{\vy}} \left(\sum_{\ind = 1}^{d} \vtheta[\ind ] \log \vx[\ind ]\right)
        \\
        & = (\alpha - 1) D_{\log} (\summ{\vz} \,\big\|\, \summ{\vy}) + \frac{\summ{\vz}}{\summ{\vy}} \kl{\vtheta}{\vx} + \left(1 - \frac{\summ{\vz}}{\summ{\vy}}\right) \ent{\vtheta} - \left(1 - \frac{\summ{\vz}}{\summ{\vy}}\right) \ent{\vx},
    \]
    where $D_{\log}$ is the Bregman divergence induced by log regularizer. Thus, the proof is concluded.
\end{proof}

\begin{proposition} \label{prop:el1_lifted}
    The Bregman divergence $D_\psi( . \,\big\|\,  .)$ induced by the regularizer $\psi(.)$ is lower bounded by a term proportional to the $\ell_1$ norm on the lifted simplex $(0, 1]\Delta^d$, 
    \[
    D_\psi(  \vy \,\big\|\,  \vz) & \geq \frac{1}{2\eta} \| \vy -  \vz\|_1^2.
    \]
\end{proposition}
\begin{proof}
    We first show that $\psi(\vy) $ is strongly convex w.r.t. $\ell_1$ norm. By \Cref{theorem:regularizer_convex}, for any vector $\nut \in \bbR^d$,
    \[
    \nut^\top \nabla^2 \psi(\vy) \nut & \geq \frac{1}{2\eta} \sum_{i = 1}^{d} \frac{\nut[i]^2}{\vy[i].\sum_{\ind = 1}^{d} \vy[\ind]} \\
    & \geq \frac{1}{2\eta} \sum_{i = 1}^{d} \frac{\nut[i]^2}{\vy[i]} \numberthis{eq:bregman_to_l11} \\
    & \geq \frac{1}{2\eta} \Big(\sum_{\ind = 1}^{d} \vy[\ind]\Big). \sum_{i = 1}^{d} \frac{\nut[i]^2}{\vy[i]} \numberthis{eq:bregman_to_l12}
     \\
    & \geq \frac{1}{2\eta} \Big(\sum_{i = 1}^{d} \nut[i]\Big)^2
    = \frac{1}{2\eta} \| \nut \|_1^2,
    \] 
    where \Cref{eq:bregman_to_l11,eq:bregman_to_l12} follow since $\sum_{\ind = 1}^{d} \vy[\ind] \leq 1$ and the last line is derived by Cauchy–Schwarz. 
    Next by definition of the Bregman divergence and strong convexity of $\psi$, 
    \[
     D_\psi(  \vy \,\big\|\,  \vz) & = \psi(\vy) - \psi(\vz) - \langle \nabla \psi (\vz), \vy - \vz \rangle \\
     & \geq \frac{1}{2\eta} \| \vy -  \vz\|_1^2.
    \]
    and the proof is completed.
\end{proof}

\begin{proposition}\label{prop:el1_simplex}
    Under the multiplicative stability assumption of $\omega \defeq \frac{\summ{\vz}}{\summ{\vy}} \in [1 - \epsilon, 1 + \epsilon]$ for a constant $\epsilon \in (0, \frac{2}{5})$, the Bregman divergence $D_\psi(\cdot \,\big\|\, \cdot)$ induced by the regularizer $\psi(\cdot)$ is lower bounded by a term proportional to the $\ell_1$ norm on the action simplex $\Delta^d$:
    \[
        \eta D_{\psi}(\vz \,\big\|\, \vy) \geq \frac{1}{4} (1 - \epsilon) \|\vtheta - \vx \|_1^2.
    \]
\end{proposition}
\begin{proof}
    By \Cref{prop:bregman_rep} and the multiplicative stability assumption $\omega \defeq \frac{\summ{\vz}}{\summ{\vy}} \in [1 - \epsilon, 1 + \epsilon]$, we infer that
    \[
        \eta D_{\psi}(\vz \,\big\|\, \vy) & \geq \beta \log^2 d \left(\log\left(\frac{1}{\omega}\right) + \omega - 1\right) + (1 - \omega) \big(\ent{\vtheta} - \ent{\vx}\big) + \omega \kl{\vtheta}{\vx}                                                                       \\
        & \geq \frac{1}{4} \beta \log^2 d \left(1 - \frac{1}{\omega}\right)^2 + (\omega - 1) \log d \sqrt{2 \kl{\vtheta}{\vx}} + \frac{\omega^2}{\beta} \kl{\vtheta}{\vx} + (\omega - \frac{\omega^2}{\beta}) \kl{\vtheta}{\vx} \numberthis{eq:ent_to_kl} \\
        & \geq \mleft(\sqrt{\beta} \log d \left(1 - \frac{1}{\omega}\right) + \frac{\omega }{\sqrt{\beta}} \sqrt{\kl{\vtheta}{\vx}} \mright)^2 + \frac{\omega}{2} \kl{\vtheta}{\vx}                                                                       \\
        & \geq \frac{1}{4} (1 - \epsilon) \|\vtheta- \vx \|_1^2,
    \]
    where \Cref{eq:ent_to_kl} follows from the inequality $| \ent{\vec p} - \ent{\vec q} | \leq (\log d) \sqrt{2 \kl{\vec p}{\vec q}}$ for discrete probability distributions $\vec p$ and $\vec q$ with support size $d$ (\Cref{lemma:entropy_diff}), and the fact that $\omega - \frac{\omega^2}{\beta} \geq \frac{\omega}{2}$ for all choices of $\omega \in [1 - \epsilon, 1 + \epsilon]$ and $\beta \geq 20$. The last line follows from Pinsker's inequality, i.e., $\| \vec p - \vec q \|_1^2 \leq 2 \kl{\vec p}{\vec q}$ for discrete random variables $\vec p$ and $\vec q$ (\Cref{lemma:pinsker}).
\end{proof}

\subsubsection{Positive Regret} \label{sec:pos_regret}
To analyze \ours equivalently as shown in \Cref{sec:equivalent_viewpoints}, we start the analysis by a closer look at \Cref{algo:ours_yview}. To analyze the $\reg^T$, we first study the nonnegative regret,
\begin{align*}
    \tildereg\^T \defeq \max_{\vy^* \in [0,1]\Delta^d } \sum_{t = 1}^{T} \langle  {\ut}\^{t}, \vy^* - \vy\^{t} \rangle.
\end{align*}
\begin{restateproposition}{prop:reg+}
    For any time horizon $T \in \mathbb{N}$, we have that $\tildereg\^T = \max\{0, \reg\^T\}$. As a result, $\tildereg\^T \geq 0$ and $\tildereg\^T \geq \reg\^T$.
\end{restateproposition}
\begin{proof}\allowdisplaybreaks
    By definition of the reward signal $ {\ut}\^t = \nut\^t -\langle \nut\^t, \vx\^t\rangle \vec1_d$ and induced action $\vx\^t = \frac{\vy\^t}{\langle \vy\^t, \vec 1 \rangle}$, for the regret we infer
    \[
        \tildereg\^T & = \max_{\vy^* \in [0,1]\Delta^d } \sum_{t = 1}^{T} \langle  {\vec\ut}\^{t}, \vy^* - \vy\^{t} \rangle                                                                                                                           \\
        & = \max_{\vy^* \in [0,1]\Delta^d } \sum_{t = 1}^{T} \langle \nut\^t -\langle \nut\^t, \frac{\vy\^t}{\langle \vy\^t, \vec 1 \rangle} \rangle \vec1_d, \vy^* - \vy\^{t} \rangle                                                   \\
        & =  \max_{\vy^* \in [0,1]\Delta^d } \sum_{t = 1}^{T} \langle \nut\^t, \vy^* \rangle - \mleft\langle \langle \nut\^t, \frac{\vy\^t}{\langle \vy\^t, \vec 1 \rangle} \rangle \vec1_d, \vy^* \mright\rangle \numberthis{eq:reg+_1} \\
        & \geq \max_{\vy^* \in \Delta^d } \sum_{t = 1}^{T} \langle \nut\^t, \vy^* \rangle - \langle \nut\^t, \vx\^{t} \rangle . \langle \vec1_d, \vy^* \rangle                                                                           \\
        & \geq \max_{\vy^* \in \Delta^d } \sum_{t = 1}^{T} \langle \nut\^t, \vy^* \rangle - \langle \nut\^t, \vx\^{t} \rangle                                                                                                            \\
        & = \reg\^T,
    \]
    where \cref{eq:reg+_1} follows because of orthogonality $\vec \ut\^{t} \perp \vy\^{t}$. On the other hand, it is clear that $  \tildereg\^T \geq 0$ by choosing $0$ as the comparator.
\end{proof}

This proposition is important as it implies that any RVU bounds on $\tildereg\^T$ directly translate into nonnegative RVU bounds on $\reg^T$.

\subsubsection{Proofs for RVU Bounds (\Cref{sec:proof_sketch})} \label{sec:proof_rvu}

We state and prove the following standard lemma from Optimistic \FTRL analysis.

\begin{restatelemma}{lemma:basi_OFTRL}
    For any $ {\vy} \in \Omega$, the sequences $\{ {\vy}\^{t}\}_{t = 1}^{T}$ generated by \Cref{eq:Fdef}, it holds that
    \[
        \sum_{t = 1}^{T} \mleft \langle  {\vy} -  {\vy}\^{t},  {\ut}\^{t} \mright \rangle & \leq \psi( {\vy}) - \psi( {\vy}\^{1}) + \sum_{t = 1}^{T} \mleft \langle  {\vz}\^{t+1} -  {\vy}\^{t},  {\ut}\^{t} -   {\ut}\^{t - 1}\mright \rangle \\
        & \qquad - \sum_{t = 1}^{T} \mleft( D_{\psi}( {\vy}\^{t} \,\big\|\,  {\vz}\^{t}) + D_{\psi}( {\vz}\^{t+1} \,\big\|\,  {\vy}\^{t}) \mright).
    \]
\end{restatelemma}
\begin{proof}
    By \Cref{lemma:opt_bregman}, and optimality of $ {\vz}\^{t}$,
    \[
        G_t( {\vz}\^{t}) & \leq G_t( {\vy}\^{t}) - D_{\psi} ( {\vy}\^{t} \,\big\|\,  {\vz}\^{t})                                                  \\
        & \leq F_t( {\vy}\^{t}) + \langle  {\vy}\^{t} ,  {\ut}\^{t - 1} \rangle - D_{\psi} ( {\vy}\^{t} \,\big\|\,  {\vz}\^{t}).
    \]
    Similarly by optimality of $ {\vy}\^{t}$,
    \[
        F_t( {\vy}\^{t}) & \leq F_t( {\vz}\^{t+1}) - D_{\psi} ( {\vz}\^{t+1} \,\big\|\,  {\vy}\^{t})                                                                   \\
        & \leq G_{t+1}( {\vz}\^{t+1}) + \langle  {\vz}\^{t+1},  {\ut}\^{t} -  {\ut}\^{t-1} \rangle - D_{\psi} ( {\vz}\^{t+1} \,\big\|\,  {\vy}\^{t}).
    \]
    By merging the inequalities and aggregating over all $t$, we derive
    \[
        G_1( {\vz}\^{1}) & \leq G_{T+1}( {\vz}\^{T + 1}) + \sum_{t = 1}^{T} (\langle  {\vy}\^{t},  {\ut}\^{t} \rangle + \langle  {\vz}\^{t+1} -  {\vy}\^{t},  {\ut}\^{t} -  {\ut}\^{t-1}\rangle) \\
        & \qquad \quad - \sum_{t=1}^{T} ( D_{\psi} ( {\vy}\^{t} \,\big\|\,  {\vz}\^{t}) + D_{\psi} ( {\vz}\^{t+1} \,\big\|\,  {\vy}\^{t}).
    \]
    Plugging in $G_{T+1}( {\vz}\^{T + 1}) \leq - \langle  {\vy},  {\Ut}\^{T+1} \rangle + \psi(\vy)$ and $G_{1}( {\vz}\^{1}) = \psi( {\vy}\^{1})$, entails the proof,
    \[
        \sum_{t = 1}^{T} \mleft \langle  {\vy} -  {\vy}\^{t},  {\ut}\^{t} \mright \rangle & \leq \psi( {\vy}) - \psi( {\vy}\^{1}) + \sum_{t = 1}^{T} \mleft \langle  {\vz}\^{t+1} -  {\vy}\^{t},  {\ut}\^{t} -   {\ut}\^{t - 1}\mright \rangle \\
        & \qquad - \sum_{t = 1}^{T} \mleft( D_{\psi}( {\vy}\^{t} \,\big\|\,  {\vz}\^{t}) + D_{\psi}( {\vz}\^{t+1} \,\big\|\,  {\vy}\^{t}) \mright).
    \]
\end{proof}

\begin{lemma} \label{lemma:opt_bregman}
    Given any convex function $F: \Omega \rightarrow \bbR$ defined on the compact set $\Omega$, the minimizer $ {\vz}^* = \argmin_{ {\vz} \in \Omega} F( {\vz}) $ satisfies
    \[
        F( {\vz}^*) \leq F( {\vz}) - D_F( {\vz} \|  {\vz}^*) \qquad \forall \vz \in \Omega,
    \]
    where $D_F$ is the Bregman divergence induced by function $F$.
\end{lemma}
\begin{proof}
    By definition,
    \[
        F( {\vz}^*) = F( {\vz}) - \langle \nabla F( {\vz}^*),  {\vz} -  {\vz}^* \rangle - D_F( {\vz} \|  {\vz}^*) \leq F( {\vz}) - D_F( {\vz} \|  {\vz}^*),
    \]
    which follows by the first order optimality condition of $ {\vz}^*$.
\end{proof}

\begin{lemma}\label{lemma:bregman_to_l1}
    If $\eta \leq \frac{1}{50}$ and $\beta$ is large enough ($\beta \geq 70$), then
    \[
        \sum_{t = 1}^T \mleft( D_\psi(  {\vy}\^{t} \,\big\|\,  {\vz}\^{t} ) + D_\psi(  {\vz}\^{t+1}\,\big\|\,  {\vy}\^{t} )\mright) \geq \sum_{t = 1}^{T} \frac{1}{2 \eta} (\| \vy\^{t} - \vz\^{t} \|_{1}^2 + \|  {\vz}\^{t + 1} -  {\vy }\^{t} \|_{1}^2).
    \]    
\end{lemma}
\begin{proof}
    By multiple usage of \Cref{prop:el1_lifted},
    \[
        D_\psi(  {\vy}\^{t} \,\big\|\,  {\vz}\^{t} ) + D_\psi(  {\vz}\^{t+1}\,\big\|\,  {\vy}\^{t} ) \geq \frac{1}{2 \eta} \|  {\vy}\^{t} - {\vz}\^{t} \|_1^2 + \frac{1}{2 \eta} \|  {\vz}\^{t+1} - {\vy}\^{t} \|_1^2,
    \]
    and summing over $t \in [T]$ concludes the proof.
\end{proof}

\begin{lemma}\label{lemma:beta_terms}
    If $\eta \leq \frac{1}{50}$ and $\beta$ is large enough ($\beta \geq 70$), then
    \[
        \sum_{t = 1}^T \mleft( D_\psi(  {\vy}\^{t} \,\big\|\,  {\vz}\^{t} ) + D_\psi(  {\vz}\^{t+1}\,\big\|\,  {\vy}\^{t} )\mright) \geq \sum_{t = 1}^{T-1} \frac{1}{10 \eta} (\| \vx\^{t+1} - \vtheta\^{t+1} \|_{1}^2 + \|  {\vtheta}\^{t + 1} -  {\vx }\^{t} \|_{1}^2).
    \]
\end{lemma}
\begin{proof}
    By \Cref{theorem:stability_lambda} in \Cref{sec:learning_rate_control_problem}, we have multiplicative stability in the learning rate as the solution to the Dynamic Learning Rate Control Problem, i.e., $\omega \defeq \frac{\summ{\vz}}{\summ{\vy}} \in [1 - \epsilon, 1 + \epsilon]$ for $\epsilon = \frac{2}{5}$. Consequently, by \Cref{prop:el1_simplex}, we infer that
    \[
        \eta D_{\psi}(\vz, \vy) \geq \frac{1}{4} (1 - \epsilon) \|\vtheta - \vx \|_1^2,
    \]
    
    Next, by setting $\vz \defeq \vz\^{t + 1}$ and $\vy \defeq \vy\^{t}$, we obtain
    \[
        D_{\psi}(\vz\^{t + 1}, \vy\^{t}) \geq \frac{3}{20 \eta} \| \vtheta\^{t + 1} - \vx\^{t} \|_{1}^2 > \frac{1}{10 \eta} \| \vtheta\^{t + 1} - \vx\^{t} \|_{1}^2.
    \]
    
    Similarly,
    \[
        D_{\psi}(\vy^{t+1}, \vz^{t+1}) \geq \frac{3}{20 \eta} \| \vx\^{t+1} - \vtheta\^{t+1} \|_{1}^2 > \frac{1}{10 \eta} \| \vx\^{t+1} - \vtheta\^{t+1} \|_{1}^2.
    \]
    
    To conclude,
    \[
        \sum_{t = 1}^T \mleft( D_\psi( \vy\^{t} \,\big\|\, \vz\^{t} ) + D_\psi( \vz\^{t+1} \,\big\|\, \vy\^{t} )\mright) & \geq \sum_{t = 1}^{T-1} \mleft( D_\psi( \vy\^{t+1} \,\big\|\, \vz\^{t+1} ) + D_\psi( \vz\^{t+1} \,\big\|\, \vy\^{t} )\mright) \\
        & \geq \sum_{t = 1}^{T-1} \frac{1}{10 \eta} (\| \vx\^{t+1} - \vtheta\^{t+1} \|_{1}^2 + \| \vtheta\^{t + 1} - \vx\^{t} \|_{1}^2).
    \]
\end{proof}

\begin{lemma}\label{lemma:u_correction}
    Assuming that $\| \nut\^{t} \|_\infty \leq 1$ is satisfied for all $t \in [T]$, we have
    \[
        \| \ut\^{t} - \ut\^{t - 1} \|_{\infty}^2 \leq 6 \|\nut\^t - \nut\^{t - 1}\|_{\infty}^2 + 4 \| \vx\^t - \vx\^{t-1} \|_1^2.
    \]
\end{lemma}
\begin{proof}
    \[
        \hspace{-0.2cm} \| {\ut}\^{t} -   {\ut}\^{t - 1} \|_{\infty}^2 & = \| ({\nut}\^{t} - \langle \nut\^t, \vx\^t \rangle \vec{1} ) -   ({\nut}\^{t-1} - \langle \nut\^{t-1}, \vx\^{t-1} \rangle \vec{1} ) \|_{\infty}^2 \\
        & \leq \mleft(\|\nut\^t - \nut\^{t - 1}\|_{\infty} + \| \langle \nut\^t, \vx\^t \rangle \vec{1} - \langle \nut\^{t-1}, \vx\^{t-1} \rangle \vec{1} \|_{\infty} \mright)^2 \numberthis{eq:u_correction1} \\
        & =  \mleft(\|\nut\^t - \nut\^{t - 1}\|_{\infty} + \big| \langle \nut\^t, \vx\^t \rangle  - \langle \nut\^{t-1}, \vx\^{t-1} \rangle \big|\mright)^2 \\
        & \leq 2 \|\nut\^t - \nut\^{t - 1}\|_{\infty}^2 + 2 \big| \langle \nut\^t, \vx\^t \rangle  - \langle \nut\^{t-1}, \vx\^{t-1} \rangle \big|^2 \numberthis{eq:u_correction2} \\
        & \leq  2 \|\nut\^t - \nut\^{t - 1}\|_{\infty}^2 + 2 \big| \big(\langle \nut\^t, \vx\^t \rangle  - \langle \nut\^t, \vx\^{t-1} \rangle\big) + \big(\langle \nut\^t, \vx\^{t-1} \rangle - \langle \nut\^{t-1}, \vx\^{t-1} \rangle\big) \big|^2 \\
        & \leq  2 \|\nut\^t - \nut\^{t - 1}\|_{\infty}^2 + 4 \big|\langle \nut\^t, \vx\^t - \vx\^{t-1}\rangle\big|^2 + 4 \big| \langle \nut\^t - \nut\^{t-1}, \vx\^{t-1} \rangle \big|^2 \numberthis{eq:u_correction3} \\
        & \leq  2 \|\nut\^t - \nut\^{t - 1}\|_{\infty}^2 + 4 \| \vx\^t - \vx\^{t-1} \|_1^2 + 4 \| \nut\^t - \nut\^{t-1} \|_{\infty}^2 \numberthis{eq:u_correction4} \\
        & = 6 \|\nut\^t - \nut\^{t - 1}\|_{\infty}^2 + 4 \| \vx\^t - \vx\^{t-1} \|_1^2,
    \]
    where \Cref{eq:u_correction1} uses the triangle inequality, \Cref{eq:u_correction2,eq:u_correction3} apply Young’s inequality, and \Cref{eq:u_correction4} utilizes Hölder’s inequality.
\end{proof}

\begin{restate}{theorem:rvu}[RVU bound of \ours]
    Consider the cumulative regret $\tildereg\^T$ accrued by the internal \OFTRL algorithm up to time $T$. Assuming that $\|  {\nut}\^{t} \|_\infty \leq 1$ is satisfied for all $t \in [T]$, it follows that for any time $T \in \bbN$ and any learning rate $\eta \leq \frac{1}{50}$ and $\beta$ high enough ($\beta \geq 70$),
    \[
        \tildereg\^T \leq 3 + \frac{\alpha \log T + \log d}{\eta} + 6 \eta \sum_{t=1}^{T-1} \| \nut\^{t} - \nut\^{t-1} \|_{\infty}^2 - \frac{1}{24\eta} \sum_{t=1}^{T-1} \| \vx\^{t+1} - \vx\^{t}  \|_{1}^2.
    \]
\end{restate}
\begin{proof}
    For any choice of comparator $ {\vy} \in \Omega$, let $ {\vy}^\prime = \frac{T - 1}{T} {\vy} + \frac{1}{T}  {\vy}\^{1} \in \Omega$. Recall that $ {\vy}\^{1} = \argmin_{\vy \in \Omega} F_1(\vy) = \argmin_{\vy \in \Omega} \psi(\vy)$. By straightforward calculations,
    \[
        \sum_{t = 1}^T  \mleft\langle 
        {\vy} -  {\vy}\^{t},  {\ut}\^{t} \mright\rangle & = \sum_{t = 1}^T \langle  {\vy} -  {\vy}^\prime,  {\ut}\^{t} \rangle + \sum_{t = 1}^T \langle  {\vy}^\prime -  {\vy}\^{t},  {\ut}\^{t} \rangle           \\
        & = \frac{1}{T} \sum_{t = 1}^T \langle  {\vy} -  {\vy}\^{1},  {\ut}\^{t} \rangle + \sum_{t = 1}^T \langle  {\vy}^\prime -  {\vy}\^{t},  {\ut}\^{t} \rangle \\
        & \leq 2 + \sum_{t = 1}^T \langle  {\vy}^\prime -  {\vy}\^{t},  {\ut}\^{t} \rangle,
    \]
    where the last line follows because of Hölder's inequality and $\|  {\ut}\^{t} \|_\infty \leq 1$.
    In turn, we need to upperbound the $\sum_{t = 1}^T \langle  {\vy}^\prime -  {\vy}\^{t},  {\ut}\^{t} \rangle$ term. By \Cref{lemma:basi_OFTRL},
    \[
        \sum_{t = 1}^{T} \mleft \langle  {\vy}^\prime -  {\vy}\^{t},  {\ut}\^{t} \mright \rangle & \leq \underbrace{\psi( {\vy}^\prime) - \psi( {\vy}\^{1})}_{(\textup{I})} + \underbrace{\sum_{t = 1}^{T} \mleft \langle  {\vz}\^{t+1} -  {\vy}\^{t},  {\ut}\^{t} -   {\ut}\^{t - 1}\mright \rangle}_{(\textup{II})} \\
        & \qquad \underbrace{- \sum_{t = 1}^{T} \mleft( D_{\psi}( {\vy}\^{t} \,\big\|\,  {\vz}\^{t}) + D_{\psi}( {\vz}\^{t+1} \,\big\|\,  {\vy}\^{t}) \mright)}_{(\textup{III})}.
    \]
    For the term (I), after some calculations,
    \[
        (\textup{I}) & =  - \frac{1}{\eta} \alpha \log \Big(\sum_{\ind=1}^d  {\vy}^\prime [\ind]\Big) + \frac{1}{\eta} \frac{1}{\sum_{\ind=1}^d  {\vy}^\prime[\ind]} \sum_{\ind=1}^d  {\vy}^\prime[\ind] \log  {\vy}^\prime[\ind]                                       \\
        & \qquad + \frac{1}{\eta} \alpha \log \Big(\sum_{\ind=1}^d  {\vy}\^{1}[\ind]\Big) - \frac{1}{\eta} \frac{1}{\sum_{\ind=1}^d  {\vy}\^{1}[\ind]} \sum_{\ind=1}^d  {\vy}\^{1}[\ind] \log  {\vy}\^{1}[\ind]                                            \\
        & \leq - \frac{1}{\eta} \alpha \log \Bigg(\frac{\sum_{\ind=1}^d  {\vy}^\prime [\ind]}{\sum_{\ind=1}^d  {\vy}\^{1}[\ind]}\Bigg) + \frac{1}{\eta} \frac{1}{\sum_{\ind=1}^d  {{\vy}^\prime}[\ind]} \sum_{\ind=1}^d  {{\vy}^\prime}[\ind] \log  {{\vy}^\prime}[\ind] \\
        & \leq \frac{1}{\eta} (\alpha \log T + \log d)
    \]
    By Hölder's and Young's inequalities, we entail that term (II) is upper bounded by
    \[
        (\textup{II}) & \leq \sum_{t = 1}^{T} \mleft \langle  {\vz}\^{t+1} -  {\vy}\^{t},  {\ut}\^{t} -   {\ut}\^{t - 1}\mright \rangle \\
        & \leq \sum_{t = 1}^{T} \| {\vz}\^{t+1} -  {\vy}\^{t}\|_{1} \cdot \| {\ut}\^{t} -   {\ut}\^{t - 1} \|_{\infty}  \\
        & \leq \sum_{t = 1}^{T} \mleft( \frac{1}{4 \eta} \| {\vz}\^{t+1} -  {\vy}\^{t}\|_{1}^2 + \eta \| {\ut}\^{t} -   {\ut}\^{t - 1} \|_{\infty}^2 \mright) \\
        & \leq \sum_{t = 1}^{T} \mleft(\frac{1}{4 \eta} \| {\vz}\^{t+1} -  {\vy}\^{t}\|_{1}^2 + 6 \eta \|\nut\^t - \nut\^{t - 1}\|_{\infty}^2 + 4 \eta \| \vx\^t - \vx\^{t-1} \|_1^2 \mright),
    \]
    where we used \Cref{lemma:u_correction}. In turn, for term (III),
    \[
        (\textup{III}) & = - \frac{1}{2} \sum_{t = 1}^{T} \mleft( D_{\psi}( {\vy}\^{t} \,\big\|\,  {\vz}\^{t}) + D_{\psi}( {\vz}\^{t+1} \,\big\|\,  {\vy}\^{t}) \mright) - \frac{1}{2} \sum_{t = 1}^{T} \mleft( D_{\psi}( {\vy}\^{t} \,\big\|\,  {\vz}\^{t}) + D_{\psi}( {\vz}\^{t+1} \,\big\|\,  {\vy}\^{t}) \mright) \\
        & \leq  - \sum_{t = 1}^{T} \frac{1}{4 \eta} (\| \vy\^{t} - \vz\^{t} \|_{1}^2 + \|  {\vz}\^{t + 1} -  {\vy }\^{t} \|_{1}^2) - \sum_{t = 1}^{T-1} \frac{1}{20 \eta} (\| \vx\^{t+1} - \vtheta\^{t+1} \|_{1}^2 + \|  {\vtheta}\^{t + 1} -  {\vx }\^{t} \|_{1}^2) \numberthis{eq:termIII} \\
        & \leq  - \sum_{t = 1}^{T} \frac{1}{4 \eta} (\| \vy\^{t} - \vz\^{t} \|_{1}^2 + \|  {\vz}\^{t + 1} -  {\vy }\^{t} \|_{1}^2) - \sum_{t = 1}^{T-1} \frac{1}{20 \eta} \| \vx\^{t+1} - \vx\^{t} \|_{1}^2,
    \]
    where \Cref{eq:termIII} is obtained by applying \Cref{lemma:beta_terms,lemma:bregman_to_l1} and the last line is yielded by triangle inequality. Assembling the complete picture, 
    \[
        \textup{(II)} + \textup{(III)} & \leq \sum_{t=1}^T 6 \eta \| {\nut}\^{t} -  {\nut}\^{t - 1} \|_{\infty}^2 + \sum_{t=1}^T 4 \eta \| \vx\^t - \vx\^{t-1} \|_1^2 - \sum_{t = 1}^{T-1} \frac{1}{20 \eta} \| \vx\^{t+1} - \vx\^{t} \|_{1}^2 \\
        & \leq 24 \eta + \sum_{t=1}^{T-1} 6 \eta \| {\nut}\^{t} -  {\nut}\^{t - 1} \|_{\infty}^2 + 16 \eta - \sum_{t = 1}^{T-1} \big(\frac{1}{20 \eta} - 4 \eta \big) \| \vx\^{t+1} - \vx\^{t} \|_{1}^2 \\
        & \leq 1 + \sum_{t=1}^{T-1}  6 \eta  \| {\nut}\^{t} -  {\nut}\^{t - 1} \|_{\infty}^2 - \sum_{t = 1}^{T-1} \frac{1}{24 \eta} \| \vx\^{t+1} - \vx\^{t} \|_{1}^2.
    \]
    and this completes the proof.
\end{proof}

\subsubsection{Proofs for Main Results} \label{sec:main_results_proofs}

\begin{restate}{theorem:bound_on_path_length}[Bound on total path length]
    Under \Cref{assumption:bound}, if all the players follow \ours algorithm with learning rate $\eta \leq \min\{\frac{1}{50}, \frac{1}{12 \sqrt{2} L n} \}$, then
    \[
    \sum_{i = 1}^{n} \sum_{t = 1}^{T - 1} \|\vx\^{t+1}_i - \vx\^{t}_i \|_1^2 \leq 144 n \eta + 48 n (\alpha \log T + \log d).
    \]
\end{restate}
\begin{proof}
    \Cref{assumption:bound} implies that,
    \[
    \| \nut\^{t + 1}_i - \nut\^{t}_i \|_{\infty}^2 & \leq L^2 \mleft(\sum_{i = 1}^{n} \| \vx\^{t+1}_i - \vx\^{t}_i \|_1 \mright)^2 \\
    & \leq L^2 n \sum_{i = 1}^{n} \| \vx\^{t+1}_i - \vx\^{t}_i \|_1^2,
    \]
    where the last line is obtained by Jensen's inequality. Next, we combine this result with the RUV bound on $\tildereg\^T$ for the $i$th player (\Cref{theorem:rvu}),
    \[
    \tildereg\^T_i & \leq 3 + \frac{\alpha \log T + \log d}{\eta} + 6 \eta \sum_{t=1}^{T-1} \| \nut\^{t + 1}_i - \nut\^{t}_i \|_{\infty}^2 - \frac{1}{24\eta} \sum_{t=1}^{T-1} \| \vx\^{t+1}_i - \vx\^{t}_i  \|_{1}^2 \\
    & \leq 3 + \frac{\alpha \log T + \log d}{\eta} + (6 L^2 n)  \eta \sum_{j = 1}^{n} \sum_{t=1}^{T-1} \| \vx\^{t+1}_j - \vx\^{t}_j \|_1^2 - \frac{1}{24\eta} \sum_{t=1}^{T-1} \| \vx\^{t+1}_i - \vx\^{t}_i  \|_{1}^2.
    \]
    Summing over all the players $i \in [n]$,
    \[
    \sum_{i = 1}^{n}  \tildereg\^T_i & \leq 3 n +  n \frac{\alpha \log T + \log d}{\eta} + \mleft( 6 L^2 n^2 \eta - \frac{1}{24\eta} \mright) \sum_{j = 1}^{n} \sum_{t=1}^{T-1} \| \vx\^{t+1}_j - \vx\^{t}_j \|_1^2 \\
    & \leq 3 n +  n \frac{\alpha \log T + \log d}{\eta} - \frac{1}{48\eta} \sum_{j = 1}^{n} \sum_{t=1}^{T-1} \| \vx\^{t+1}_j - \vx\^{t}_j \|_1^2,
    \]
    since $\eta^2 \leq \frac{1}{288 L^2 n^2}$. Now, by recalling that $\tildereg\^T_i \geq 0$, we get,
    \[
    & 0 \leq 3 n +  n \frac{\alpha \log T + \log d}{\eta} - \frac{1}{48\eta} \sum_{j = 1}^{n} \sum_{t=1}^{T-1} \| \vx\^{t+1}_j - \vx\^{t}_j \|_1^2,
    \]
    implying
    \[
        \sum_{j = 1}^{n} \sum_{t=1}^{T-1} \| \vx\^{t+1}_j - \vx\^{t}_j \|_1^2 \leq 144 n \eta + 48 n (\alpha \log T + \log d).
    \]
\end{proof}

\begin{restate}{theorem:regret_bound}[Regret bound of \ours]
    Under \Cref{assumption:bound}, if all the players $i \in [n]$ follows \ours with learning rate $\eta = \min\{\frac{1}{50}, \frac{1}{12 \sqrt{2} L n} \}$, then the regret of each player $i \in [n]$ is bounded as,
    \[
    \reg_i\^T \leq 6 + \max\{50 + 12 \sqrt{2} L n, 24 \sqrt{2} L n\} (\alpha \log T + \log d),
    \]
    and the algorithm for each player $i \in [n]$ is adaptive to adversarial utilities, i.e., the regret that each player incurs is $\reg_i\^T = \Tilde{O}(\sqrt{T \log d})$.
\end{restate}
\begin{proof}
    Similar to the proof of \Cref{theorem:bound_on_path_length},
    \[
    \| \nut\^{t + 1}_i - \nut\^{t}_i \|_{\infty}^2 & \leq L^2 \mleft(\sum_{i = 1}^{n} \| \vx\^{t+1}_i - \vx\^{t}_i \|_1 \mright)^2 \\
    & \leq L^2 n \sum_{i = 1}^{n} \| \vx\^{t+1}_i - \vx\^{t}_i \|_1^2.
    \]
    Summing over $t$ from $1$ to $T-1$,
    \[
    \sum_{t = 1}^{T-1} \| \nut\^{t + 1}_i - \nut\^{t}_i \|_{\infty}^2 & \leq  L^2 n \sum_{t = 1}^{T-1} \sum_{i = 1}^{n} \| \vx\^{t+1}_i - \vx\^{t}_i \|_1^2, \\
    & \leq L^2 n \mleft( 144 n \eta + 48 n (\alpha \log T + \log d) \mright) \\
    & = 144 L^2 n^2 \eta + 48 L^2 n^2 (\alpha \log T + \log d) \numberthis{eq:check},
    \]
    where we leveraged \Cref{theorem:bound_on_path_length}. By \Cref{prop:reg+} and \Cref{theorem:rvu} we infer that,
    \[
    \reg_i\^T & \leq \tildereg_i\^T \\
    & \leq 3 + \frac{\alpha \log T + \log d}{\eta} + 6 \eta \sum_{t=1}^{T-1} \| \nut\^{t} - \nut\^{t-1} \|_{\infty}^2 - \frac{1}{24\eta} \sum_{t=1}^{T-1} \| \vx\^{t+1} - \vx\^{t}  \|_{1}^2 \\
    & \leq 3 + \frac{\alpha \log T + \log d}{\eta} + 6 \eta \sum_{t=1}^{T-1} \| \nut\^{t} - \nut\^{t-1} \|_{\infty}^2 \\
    & \leq 3 + \frac{\alpha \log T + \log d}{\eta} + 864 L^2 n^2 \eta^2 + 288 L^2 n^2 \eta (\alpha \log T + \log d) \numberthis{eq:regret_bound1} \\
    & \leq 6 + \frac{\alpha \log T + \log d}{\eta} + 12 \sqrt{2} L n (\alpha \log T + \log d) \\
    & \leq 6 + \max\{50, 12 \sqrt{2} L n\} (\alpha \log T + \log d) + 12 \sqrt{2} L n (\alpha \log T + \log d) \\
    & \leq 6 + \max\{50 + 12 \sqrt{2} L n, 24 \sqrt{2} L n\} (\alpha \log T + \log d),
    \]
    where \Cref{eq:regret_bound1} is due to \Cref{theorem:bound_on_path_length}, and the last lines are because $\eta = \min\{\frac{1}{50}, \frac{1}{\sqrt{288} L n} \}$.
    
    To prove the adversarial bound for each player $i \in [n]$, player $i$ simply check if there exists a time $t \in [T]$,  such that the  
    \[
    \sum_{\tau = 1}^{t-1} \| \nut\^{\tau + 1}_i - \nut\^{\tau}_i \|_{\infty}^2 & >  144 L^2 n^2 \eta + 48 L^2 n^2 (\alpha \log t + \log d),
    \]
    and if noticed that, start to switch to any no-regret learning algorithm, e.g., \Hedge \citep{Cesa-Bianchi06:Prediction} and get $O(\sqrt{T \log d})$ regret. The argument is based on the fact that if all the players follow the \ours dynamics, then \Cref{eq:check} should be satisfied.
\end{proof}

\section{Conclusion}

We introduced an uncoupled online learning algorithm that achieves near-constant regret of $O(n \log^2 d \log T)$ in multi-player general-sum games. This significantly improves upon the $O(d \log T)$ regret achieved by Log-Regularized Lifted Optimistic FTRL, exponentially reducing the dependence on the number of actions $d$~\citep{farina2022near}. Furthermore, our algorithm reduces the dependence on the number of iterations $T$ from $O(\log^4 T)$ in the Optimistic Hedge algorithm to $O(\log T)$, improving upon the regret bound of $O(n \log d \log^4 T)$~\citep{daskalakis2021near}. At the heart of these improvements lies a dynamic, nonmonotonic pacing of the learning rate. Specifically, players slow down their learning when their regret becomes too negative—that is, when they are significantly outperforming all fixed actions.

While our algorithm achieves near-constant regret guarantees, it remains an interesting open question whether constant regret is achievable for regularized learning in general games. Another natural direction for future research is to explore how our ideas can be applied more broadly across regularized learning algorithms. Our adaptive learning rate framework may be fruitfully combined with other FTRL-based or OMD-based methods—beyond Optimistic Multiplicative Weight Updates—to push the boundaries of performance and potentially provide a unified perspective on accelerated no-regret learning in games.

Additionally, dynamic learning rate ideas could prove valuable in minimizing stronger notions of regret, such as swap regret, within game-theoretic settings. Beyond regret minimization, a particularly compelling challenge lies in understanding the \emph{day-to-day} dynamics of learning with adaptive pacing in structured games—offering a finer-grained view of convergence behavior and opening the door to new theoretical insights and practical strategies.

More broadly, this work contributes to a shift in multi-agent learning: rather than relying on prespecified schedules for learning rates (e.g., fixed or monotonically decreasing steps such as $1/\sqrt{t}$), we advocate for dynamically adaptive learning rates that respond to real-time performance. Although step-size tuning is widely recognized as a critical component of single-agent learning—particularly in neural network training (see, e.g.,~\citet{bengio2012practical})—such considerations have received far less attention in multi-agent and game-theoretic settings.

We hope that our work stimulates further discussion and research at this intersection. In particular, we believe that developing game-aware adaptive schemes opens up a rich and open-ended research direction—one that bridges online learning, control theory, dynamical systems, and game theory, and may ultimately lead to fundamentally new learning dynamics tailored to strategic multi-agent environments.

\section{Acknowledgments}

The authors appreciate Patrick Jaillet for his insightful comments and valuable suggestions. A.S. was partially supported by the National Research Foundation Singapore and DSO National Laboratories under the AI Singapore Programme AISG Award No: AISG2-RP-2020-018, and by the Office of Naval Research (ONR) grant N00014-24-1-2470. G.F. acknowledges the support of NSF Award CCF-244306.

\printbibliography

\newpage 

\appendix

\end{document}